\documentclass[journal,comsoc]{IEEEtran}

\usepackage{balance}
\usepackage{amssymb,amsthm,amsmath,epsfig,latexsym,graphicx,bm,bbm,nccmath}
\allowdisplaybreaks[4] 
\usepackage{mathtools}

\usepackage{epstopdf}

\usepackage{booktabs}
\usepackage{tabularx}

\usepackage{tikz}
\usetikzlibrary{arrows.meta,positioning,fit,backgrounds,calc}

\usepackage{enumerate}
\usepackage{cite}
\usepackage{amsfonts}
\usepackage[ruled, vlined, linesnumbered, commentsnumbered, longend]{algorithm2e}
\SetKw{KwBreak}{break}

\usepackage{textcomp}
\usepackage{xcolor}
\def\BibTeX{{\rm B\kern-.05em{\sc i\kern-.025em b}\kern-.08em
		T\kern-.1667em\lower.7ex\hbox{E}\kern-.125emX}}
\usepackage[utf8x]{inputenc}
\usepackage{etoolbox}
\usepackage{multirow}
\usepackage[caption=false]{subfig}

\newtheorem{theorem}{Theorem}
\newtheorem{proposition}{Proposition}
\newtheorem{lemma}{Lemma}

\newtheorem{definition}{Definition}
\newtheorem{assumption}{Assumption}
\newtheorem{remark}{Remark}
\usepackage{easyReview} 

\makeatletter
\def\@IEEEfigurecaptionsepspace{\vskip 0pt\relax}
\def\@IEEEtablecaptionsepspace{\vskip 0pt\relax}
\makeatother

\begin{document}
	\bstctlcite{IEEEexample:BSTcontrol}
	\title{Joint Antenna Positioning and Pilot Power Allocation for Uplink Fluid Antenna-Aided Cell-Free Massive MIMO Systems}
	
	\author{Manobendu~Sarker and 
		Soumaya Cherkaoui,~\IEEEmembership{Senior Member,~IEEE}
		\thanks{M. Sarker and S. Cherkaoui are with the Department of Computer and Software Engineering, Polytechnique Montr{\'{e}}al, Canada. (Corresponding author's e-mail: manobendu.sarker@polymtl.ca)
	}}
	
	\maketitle
	
	
	\begin{abstract}
		This paper investigates joint fluid antenna (FA) position optimization and pilot power allocation for uplink cell-free massive MIMO networks supporting ultra-reliable low-latency communication (URLLC) traffic under pilot contamination.
		By continuously adjusting FA elements within a prescribed aperture,
		we adapt the spatial correlation structure to the varying propagation environment.
		We formulate a joint problem to minimize the average normalized mean-squared error (NMSE) subject to quality-of-service constraints governed by finite blocklength and Markov-modulated Poisson process queueing analyses.
		To ensure tractability, we establish an analytical quadratic relationship between the NMSE and the effective signal-to-interference-plus-noise ratio through eigenstructure alignment factors. This mapping translates complex latency requirements into explicit NMSE thresholds, maintaining a median absolute approximation error below $0.03$~dB and less than $10\%$ relative error for over $90\%$ of practical-regime samples. Given the nonconvex nature of the formulation, we propose an alternating optimization framework. Antenna positions are iteratively updated via successive convex approximation using analytical gradients and closed-form Euclidean projections. Pilot powers are allocated through a block coordinate descent scheme that employs safeguarded Newton steps to guarantee monotonic convergence. Furthermore, we introduce an auto-calibrated regularization variant to navigate the Pareto frontier between estimation accuracy and energy consumption. Numerical results show the proposed positioning yields up to 22\% NMSE improvement over fixed uniform arrays while operating 7--8$\times$ faster than particle swarm optimization. The power allocation achieves within 0.5\%  of the global optimum at over two orders of magnitude lower 
		computational cost, providing up to 39\% pilot energy reduction with negligible estimation penalty.
	\end{abstract}
	\vspace{-6mm}
	\section{Introduction}
	\label{sec:introduction}
	
	Cell-free massive multiple-input multiple-output (CF mMIMO) networks, where geographically distributed access points (APs) coherently serve user equipments (UEs) over shared time-frequency resources, constitute a promising distributed architecture for reliable wireless connectivity~\cite{ngo2017cellfree}. Under user-centric (UC) clustering \cite{buzzi2017cellfree}, each UE is coherently served by a selected subset of nearby APs via time-division duplex (TDD) operation with channel reciprocity, and the quality of uplink (UL) channel estimates directly governs the achievable signal-to-interference-plus-noise ratio (SINR) for both UL and downlink (DL) transmission~\cite{bjornson2017massive}. Pilot contamination, caused by the lack of orthogonal pilot sequences compared to the number of UEs in the network, introduces correlated estimation errors that persist even as the number of antennas grows, establishing channel estimation as a fundamental performance bottleneck~\cite{marzetta2010noncooperative}. The spatial correlation structure at each AP, determined by the array geometry and the propagation environment, shapes both the estimation accuracy and the interference rejection capability through the observation covariance and the combining vectors, respectively.
	
	Fluid antenna (FA) and movable antenna (MA) technologies provide a spatial degree of freedom by enabling continuous repositioning of antenna elements within a prescribed aperture, thereby reshaping the spatial correlation matrices that govern channel estimation and signal combining~\cite{wong2021fluid}. In single-cell settings, position optimization has been shown to adapt inter-element spacings to the angular scattering environment, yielding substantial capacity and diversity gains over fixed-geometry arrays~\cite{zhu2024movable}, while joint position and beamforming designs exploit the position-dependent channel structure for interference suppression and spatial multiplexing~\cite{zhu2024movable_multiuser, Gao2024movable_multicast}. These results confirm that position optimization offers a mechanism complementary to conventional precoding, operating on the physical array geometry rather than the signal domain. However, the extension to distributed multi-AP architectures introduces additional design dimensions, as position optimization at each AP must account for the spatial correlation experienced by multiple UEs with distinct angular profiles and pilot-sharing relationships.
	
	{Recent works integrate FA/MA elements into CF mMIMO to exploit the 
		additional geometric degrees of freedom. For example, 
		six-dimensional surface placement via Bayesian optimization is 
		studied for average sum-rate maximization~\cite{Shi6DMA_CF}, deep 
		reinforcement learning-based antenna positioning is developed for 
		Doppler mitigation in high-mobility CF systems~\cite{DaiMovableCF_HST}, 
		joint AP-side antenna positioning and uplink power control is 
		addressed for max-min rate fairness~\cite{OlyaeeFluidCF}, and 
		antenna positions, DL beamforming, and UE association are 
		simultaneously optimized for weighted sum-rate 
		maximization~\cite{ZhuJBUM_CFMA}. These studies establish antenna reconfigurability as a powerful geometric design mechanism for CF mMIMO but primarily optimize conventional rate or spectral-efficiency objectives and do not couple position optimization with pilot-domain resource management under pilot contamination.}
	
	
	{Pilot power allocation offers a complementary mechanism for managing 
		channel-estimation quality, since each UE's pilot power controls 
		both its own estimation SNR and the pilot contamination experienced 
		by co-pilot UEs. Prior work in this direction includes successive 
		approximation to minimize the maximum estimation error~\cite{Mai2018a}, 
		closed-form allocations derived from a zero-normalized mean-squared error (NMSE) equilibrium~\cite{Sark2301:Pilot}, and quadratic-transform updates that avoid solver dependence~\cite{Sarker_RL_2026}. These schemes operate under fixed array geometries, so the interplay between position-dependent spatial correlation and pilot-domain resource management has not been considered yet.}
	
	
	{Delivering ultra-reliable low-latency communication (URLLC) traffic in this joint design imposes stringent 
		probabilistic delay and reliability constraints, which require an 
		explicit mapping from physical-layer estimation parameters to 
		queueing-level performance guarantees~\cite{sarker2025NS}. Under finite blocklength (FBL) transmission with bursty traffic modeled as a Markov-modulated Poisson process (MMPP)~\cite{fischer1993markov}, the minimum required service rate translates, through the FBL expression, into a per-UE SINR target whose satisfaction is required to meet the prescribed delay-violation probability. 
		Existing SINR expressions for CF mMIMO involve trace products of 
		channel-estimate and interference covariance matrices whose 
		dependence on positions and powers is implicit, precluding direct 
		translation of quality-of-service (QoS) requirements into actionable 
		optimization targets.}
	
	
	{These developments reveal two complementary but disconnected research directions. Existing FA/MA-aided CF mMIMO studies exploit antenna reconfigurability for rate enhancement, interference management, mobility support, or UE association, but do not jointly address pilot power allocation under URLLC-oriented QoS constraints. Conversely, pilot-resource schemes mitigate pilot contamination and channel-estimation error under fixed array geometries. Hence, the joint treatment of AP-side FA configuration, pilot power allocation, and estimation-aware URLLC guarantees remains unaddressed.}
	
	{To address this challenge, we adopt the NMSE as a unifying estimation-quality metric and establish a quadratic mapping between the per-UE NMSE and effective SINR through eigenstructure-based alignment factors. This mapping converts the FBL- and MMPP-derived SINR requirements into explicit per-UE NMSE thresholds, enabling joint antenna-position and pilot power optimization under pilot contamination. Table~\ref{tab:related_work_comparison} summarizes the feature-level distinctions between representative related works and the proposed framework.}
	\begin{table*}[t]
		
		\centering
		\footnotesize
		\caption{{Comparison with Related FA/MA-Aided CF mMIMO and Pilot-Resource Optimization Works}}
		\label{tab:related_work_comparison}
		\resizebox{\textwidth}{!}{
			\begin{tabular}{lcccccccc}
				\hline
				\textbf{Work} & \textbf{FA/MA} & \textbf{AP-side FA/MA} & \textbf{Pilot} & \textbf{Imperfect} & \textbf{Pilot pow.} & \textbf{URLLC/} & \textbf{QoS/} & \textbf{NMSE--SINR} \\
				&  & \textbf{opt.} & \textbf{cont.} & \textbf{CSI/CE/NMSE} & \textbf{opt.} & \textbf{FBL} & \textbf{delay} & \textbf{map} \\
				\hline
				\cite{OlyaeeFluidCF} & $\checkmark$ & $\checkmark$ & $\times$ & $\checkmark$ & $\times$ & $\times$ & $\times$ & $\times$ \\
				\cite{Shi6DMA_CF} & $\checkmark$ & $\checkmark$ & $\times$ & $\times$ & $\times$ & $\times$ & $\times$ & $\times$ \\
				\cite{DaiMovableCF_HST} & $\checkmark$ & $\checkmark$ & $\times$ & $\times$ & $\times$ & $\times$ & $\times$ & $\times$ \\
				\cite{ZhuJBUM_CFMA} & $\checkmark$ & $\checkmark$ & $\times$ & $\times$ & $\times$ & $\times$ & $\times$ & $\times$ \\
				\cite{Mai2018a} & $\times$ & $\times$ & $\checkmark$ & $\checkmark$ & $\checkmark$ & $\times$ & $\times$ & $\times$ \\
				\cite{Sark2301:Pilot} & $\times$ & $\times$ & $\checkmark$ & $\checkmark$ & $\checkmark$ & $\times$ & $\times$ & $\times$ \\
				\cite{Sarker_RL_2026} & $\times$ & $\times$ & $\checkmark$ & $\checkmark$ & $\checkmark$ & $\times$ & $\times$ & $\times$ \\
				\textbf{This work} & $\checkmark$ & $\checkmark$ & $\checkmark$ & $\checkmark$ & $\checkmark$ & $\checkmark$ & $\checkmark$ & $\checkmark$ \\
				\hline
		\end{tabular}}
		\vspace{1mm}
		\begin{minipage}{0.98\textwidth} \footnotesize \textbf{AP-side FA/MA opt.} = optimization of antenna positions, ports, orientations, or placements at the APs; \textbf{Imperfect CSI/CE/NMSE} = explicit modeling of imperfect channel state information (CSI) or consideration of channel-estimation error/NMSE as a design objective or performance metric; \textbf{Pilot cont.} = pilot contamination; and \textbf{Pilot pow. opt.} = pilot power optimization. \end{minipage}
		\vspace{-7mm}
	\end{table*}
	
	{As shown in Table~\ref{tab:related_work_comparison}, the proposed framework jointly incorporates AP-side FA optimization, pilot contamination, pilot power allocation, FBL transmission, and queueing-level QoS requirements.} To the best of our knowledge, no existing work provides a tractable analytical connection between channel-estimation quality and SINR guarantees that supports this joint design under MMPP-driven URLLC traffic constraints in CF mMIMO networks. The main contributions are summarized as follows.
	\begin{enumerate}
		\item We establish an analytical quadratic relationship between the per-UE NMSE and the effective UL SINR through eigenstructure-based alignment factors, yielding a mapping with quantifiable approximation error that vanishes under practical pilot SNR and correlation conditions. This mapping converts URLLC QoS requirements, specified through FBL rate expressions and MMPP-based queueing analysis based on~\cite{fischer1993markov,latouche1999matrix,lucantoni1991new}, into per-UE NMSE thresholds that decouple the physical-layer optimization from the queueing-level service specification. The resulting joint problem is nonconvex and solved via alternating optimization (AO).
		
		\item We develop a successive convex approximation (SCA) procedure for fluid antenna position optimization that constructs linearized surrogates with analytical gradient computation through the position-dependent spatial correlation model, and employs closed-form projected updates via isotonic regression. The scheme yields NMSE reduction comparable to metaheuristic benchmarks while operating at an order-of-magnitude lower computational cost.
		
		\item We propose a block coordinate descent (BCD) framework for pilot power allocation that decomposes the multi-UE optimization into scalar subproblems by exploiting pilot contamination separability, with closed-form gradient and Hessian expressions enabling safeguarded Newton updates with guaranteed monotonic convergence. A power-aware extension introduces auto-calibrated regularization that exposes a penalty-free operating region near the NMSE optimum, providing tunable Pareto-optimal operating points for energy-constrained deployments.
		
		\item Extensive numerical evaluation validates the analytical approximations and confirms that the proposed positioning and power allocation schemes consistently achieve near-optimal estimation quality at substantially reduced computational cost, while the power-aware variant enables significant pilot energy savings with negligible estimation degradation.
	\end{enumerate}

	\vspace{-4.5mm}
	\section{System Model}
	\label{sec:system_model}
	
	\subsection{Network Architecture and Fluid Antenna Model}
	\label{subsec:network_architecture}
	
	We consider a CF mMIMO network shown in Fig.~\ref{fig:FAS},  comprising $M$ geographically distributed APs that jointly serve $K$ single-antenna UEs over shared time-frequency resources \cite{ngo2017cellfree}. All APs connect to a central processing unit (CPU) through fronthaul links, enabling centralized baseband processing and coherent signal combining. The system operates under TDD with channel reciprocity exploited for DL precoding based on UL channel estimates. We adopt a UC architecture where each UE $k$ is served by a subset of APs denoted $\mathcal{M}_k \subseteq \{1, \ldots, M\}$, selected based on large-scale channel conditions \cite{buzzi2017cellfree}. Symmetrically, $\mathcal{K}_m \subseteq \{1, \ldots, K\}$ denotes the UEs served by AP $m$, with the relationship $\mathcal{M}_k = \{m : k \in \mathcal{K}_m\}$ holding by construction. This clustering approach balances performance gains from coherent processing against fronthaul overhead and computational complexity. 
	\begin{figure}[tb]
		\centering
		\includegraphics[scale=0.55]{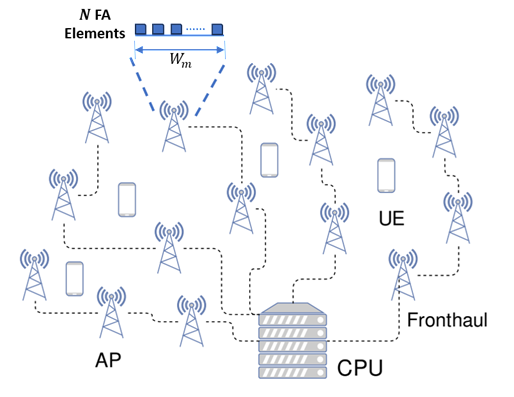}
		\caption{System diagram of the proposed FA-aided CF mMIMO network serving URLLC traffic.}
		\label{fig:FAS}
	\end{figure}
	
	Each AP $m$ is equipped with $N$ reconfigurable fluid antenna elements arranged along a linear axis of total length $W_m$ \cite{wong2021fluid}. Unlike conventional fixed arrays, fluid antennas can dynamically reposition within a prescribed region to exploit spatial channel variations. Let $\mathbf{x}_m = [x_m(1), \ldots, x_m(N)]^\mathsf{T} \in \mathbb{R}^N$ denote the position vector for AP $m$, where $x_m(n)$ specifies the coordinate of element $n$. To ensure physical realizability and limit electromagnetic coupling, the positions must satisfy
	\begin{align}
		0 \leq x_m(1) < \cdots < x_m(N) \leq W_m,  \nonumber \\
		x_m(n+1) - x_m(n) \geq d_{\min},
		\label{eq:position_constraints}
	\end{align}
	where $d_{\min} = \lambda/2$ is the minimum inter-element spacing at carrier wavelength $\lambda$. The continuous position model reflects the capability of fluid antenna technologies to occupy arbitrary locations within the movement range, and captures the smooth dependence of spatial correlation on element positions. Discrete positioning for switched-port implementations can be treated as a special case with positions restricted to a predefined grid.

	\vspace{-3mm}
	\subsection{Spatially Correlated Channel Model}
	\label{subsec:channel_model}
	
	We consider non-line-of-sight propagation where the UL channel from UE $k$ to AP $m$ is given by
	\begin{equation}
		\mathbf{h}_{k,m}(\mathbf{x}_m) = \sqrt{\beta_{k,m}} \, \mathbf{g}_{k,m}(\mathbf{x}_m) \in \mathbb{C}^N.
		\label{eq:channel_model}
	\end{equation}
	Here, $\beta_{k,m}$ captures large-scale fading effects including path loss and shadowing, which remain constant over many coherence intervals and depend on the AP-UE geometry but not on antenna positions within the array aperture. The small-scale fading vector $\mathbf{g}_{k,m}(\mathbf{x}_m) \sim \mathcal{CN}(\mathbf{0}, \mathbf{R}_{k,m}(\mathbf{x}_m))$ follows a spatially correlated Rayleigh distribution with position-dependent correlation matrix $\mathbf{R}_{k,m}(\mathbf{x}_m) \in \mathbb{C}^{N \times N}$.
	
	Following the widely adopted local scattering model \cite{adhikary2013joint, bjornson2017massive}, the spatial correlation between elements $l$ and $l'$ is
	\begin{align}
		\scalebox{1}{$
			\begin{aligned}
				[\mathbf{R}_{k,m}(\mathbf{x}_m)]_{l,l'} & =  e^{j \frac{2\pi}{\lambda} \Delta x_{l,l'}^{(m)} \sin\theta_{k,m}} \\
				& \exp\left( -\frac{1}{2} \left( \frac{2\pi \sigma_\theta}{\lambda} \Delta x_{l,l'}^{(m)} \cos\theta_{k,m} \right)^2 \right),
			\end{aligned}
			$}
		\label{eq:correlation_model}
	\end{align}
	where $\Delta x_{l,l'}^{(m)} = x_m(l) - x_m(l')$ is the signed inter-element separation, $\theta_{k,m}$ denotes the nominal angle of arrival from UE $k$ to AP $m$, and $\sigma_\theta$ characterizes the angular spread due to local scatterers around the UE. The first exponential represents the deterministic phase shift, while the second captures the decorrelation arising from angular dispersion.
	
	{Even though the adopted correlated Rayleigh model focuses on non-line-of-sight (NLoS) propagation, the framework can be extended to Rician fading \cite{10679239} by representing the channel as the sum of a position-dependent deterministic LoS response and a correlated diffuse component. When the LoS response is known or tracked as part of the long-term channel statistics, channel estimation follows a nonzero-mean linear minimum mean-squared error (LMMSE) formulation in which the deterministic mean is removed from the pilot observation before estimating the diffuse component. Since both the LoS response and the diffuse covariance depend on the antenna positions, the same statistical-CSI-based design principle remains applicable, although the subsequent analysis must account for the additional LoS-dependent terms. A complete Rician extension is left for future work.}
	
	
	\vspace{-4mm}
	\subsection{UL Pilot Transmission and Channel Estimation}
	\label{subsec:channel_estimation}
	
	Each coherence block spans $\tau_c$ symbols, partitioned into $\tau_p$ symbols for UL pilot transmission and $\tau_c - \tau_p$ symbols for data. During training, UE $k$ transmits pilot sequence $\boldsymbol{\phi}_k \in \mathbb{C}^{\tau_p}$ with $\|\boldsymbol{\phi}_k\|^2 = \tau_p$ and power $p_k^{\mathrm{p}}$. As in practical deployments, we consider $K>\tau_p$, necessitating the reuse of pilot sequences among UEs, which leads to pilot contamination. Let $\pi_k \in \{1, \ldots, \tau_p\}$ denote the pilot index assigned to UE $k$, and define the pilot-sharing set $\mathcal{P}_k = \{j \in \mathcal{K} : \pi_j = \pi_k\}$ containing all UEs using the same pilot.
	
	The received pilot signal at AP $m$, after correlation with the pilot of UE $k$, yields the sufficient statistic
	\begin{equation}
		\mathbf{y}_{k,m}^{\mathrm{p}} = \sum_{j \in \mathcal{P}_k} \sqrt{\tau_p p_j^{\mathrm{p}}} \, \mathbf{h}_{j,m}(\mathbf{x}_m) + \mathbf{n}_{k,m}^{\mathrm{p}},
		\label{eq:pilot_observation}
	\end{equation}
	where $\mathbf{n}_{k,m}^{\mathrm{p}} \sim \mathcal{CN}(\mathbf{0}, \sigma^2 \mathbf{I}_N)$ is receiver noise. The summation over $\mathcal{P}_k$ captures pilot contamination from co-pilot UEs, a fundamental limitation in massive MIMO that persists even with unlimited antennas \cite{marzetta2010noncooperative}.
	
	Applying the LMMSE criterion \cite{kay1993fundamentals}, the channel estimate is
	\begin{equation}
		\widehat{\mathbf{h}}_{k,m}(\mathbf{x}_m) = \sqrt{\tau_p p_k^{\mathrm{p}}} \, \beta_{k,m} \mathbf{R}_{k,m}(\mathbf{x}_m) \boldsymbol{\Psi}_{k,m}^{-1}(\mathbf{x}_m) \, \mathbf{y}_{k,m}^{\mathrm{p}},
		\label{eq:lmmse_estimate}
	\end{equation}
	where the observation covariance matrix is
	\begin{equation}
		\boldsymbol{\Psi}_{k,m}(\mathbf{x}_m) = \sum_{j \in \mathcal{P}_k} \tau_p p_j^{\mathrm{p}} \beta_{j,m} \mathbf{R}_{j,m}(\mathbf{x}_m) + \sigma^2 \mathbf{I}_N.
		\label{eq:psi_matrix}
	\end{equation}
	Notably, $\boldsymbol{\Psi}_{k,m}(\mathbf{x}_m)$ depends on the correlation matrices of all co-pilot UEs, coupling the estimation quality of UE $k$ to the spatial characteristics of interfering channels.
	
	The channel estimate $\widehat{\mathbf{h}}_{k,m}$ and estimation error $\widetilde{\mathbf{h}}_{k,m} = \mathbf{h}_{k,m} - \widehat{\mathbf{h}}_{k,m}$ are uncorrelated with covariance matrices
	\begin{align}
		\boldsymbol{\Phi}_{k,m}(\mathbf{x}_m) &= \tau_p p_k^{\mathrm{p}} \beta_{k,m}^2 \mathbf{R}_{k,m} \boldsymbol{\Psi}_{k,m}^{-1} \mathbf{R}_{k,m}, \label{eq:phi_matrix} \\
		\mathbf{C}_{k,m}(\mathbf{x}_m) &= \beta_{k,m} \mathbf{R}_{k,m} - \boldsymbol{\Phi}_{k,m}. \label{eq:error_covariance}
	\end{align}
	For co-pilot UEs $j \in \mathcal{P}_k$, the cross-covariance $\boldsymbol{\Phi}_{k,m}^{(j)} = \mathbb{E}\{\widehat{\mathbf{h}}_{k,m} \mathbf{h}_{j,m}^\mathsf{H}\} = \tau_p \sqrt{p_k^{\mathrm{p}} p_j^{\mathrm{p}}} \beta_{k,m} \beta_{j,m} \mathbf{R}_{k,m} \boldsymbol{\Psi}_{k,m}^{-1} \mathbf{R}_{j,m}$ quantifies the statistical coupling due to pilot sharing, which manifests as coherent interference during data detection.
	
	\vspace{-3mm}
	\subsection{UL Data Transmission}
	\label{subsec:data_transmission}
	
	During UL data transmission, UE $k$ sends symbol $s_k$ with $\mathbb{E}\{|s_k|^2\} = 1$ and power $p_k^{\mathrm{u}}$. The received signal at AP $m$ is
	\begin{equation}
		\mathbf{y}_m = \sum_{k=1}^{K} \sqrt{p_k^{\mathrm{u}}} \, \mathbf{h}_{k,m}(\mathbf{x}_m) s_k + \mathbf{n}_m,
		\label{eq:received_signal}
	\end{equation}
	where $\mathbf{n}_m \sim \mathcal{CN}(\mathbf{0}, \sigma^2 \mathbf{I}_N)$. Under centralized processing, APs forward received signals to the CPU, which performs joint detection. Using the channel estimates as combining vectors (maximum ratio combining), the CPU forms the decision statistic for UE $k$ as
	\begin{equation}
		\hat{s}_k = \sum_{m \in \mathcal{M}_k} \widehat{\mathbf{h}}_{k,m}^\mathsf{H}(\mathbf{x}_m) \, \mathbf{y}_m.
		\label{eq:mrc_detection}
	\end{equation}
	
	\vspace{-3mm}
	\subsection{UL SINR and FBL Rate}
	\label{subsec:sinr_rate}
	
	Applying the use-and-then-forget (UatF) capacity bounding technique from \cite{marzetta2016fundamentals} to \eqref{eq:mrc_detection}, after substituting \eqref{eq:received_signal} and treating channel uncertainty and interference as worst-case Gaussian noise, we obtain the effective UL SINR for UE~$k$ as
	\begin{equation}
		\mathrm{SINR}_k = \frac{\mathrm{DS}_k}{\mathrm{BU}_k + \mathrm{UI}_k + \mathrm{CI}_k + \mathrm{N}_k},
		\label{eq:sinr}
	\end{equation}
	where the signal and interference components are defined below.
	
	The desired signal power, capturing the coherent array gain, is given by
	\begin{equation}
		\scalebox{1}{$\mathrm{DS}_k = p_k^{\mathrm{u}} \left| \sum_{m \in \mathcal{M}_k} \mathrm{tr}(\boldsymbol{\Phi}_{k,m}) \right|^2$}.
		\label{eq:DS}
	\end{equation}
	
	The beamforming uncertainty arising from imperfect channel knowledge contributes
	\begin{equation}
		\scalebox{1}{$\mathrm{BU}_k = p_k^{\mathrm{u}} \sum_{m \in \mathcal{M}_k} \beta_{k,m} \, \mathrm{tr}(\boldsymbol{\Phi}_{k,m} \mathbf{R}_{k,m}).$}
		\label{eq:BU}
	\end{equation}
	
	Next, the non-coherent interference from UEs that do not share pilots with UE~$k$ is
	\begin{equation}
		\scalebox{1}{$\mathrm{UI}_k = \sum_{j \notin \mathcal{P}_k} p_j^{\mathrm{u}} \sum_{m \in \mathcal{M}_k} \beta_{j,m} \, \mathrm{tr}(\boldsymbol{\Phi}_{k,m} \mathbf{R}_{j,m}).$}
		\label{eq:UI}
	\end{equation}
	
	The coherent interference from pilot-sharing UEs, which accumulates constructively across APs, is expressed as
	\begin{align}
		\scalebox{0.95}{$
			\begin{aligned}
				\mathrm{CI}_k  
				&= \sum_{j \in \mathcal{P}_k \setminus \{k\}} p_j^{\mathrm{u}} 
				\Bigg[
				\left| \sum_{m \in \mathcal{M}_k} \mathrm{tr}\!\big(\boldsymbol{\Phi}_{k,m}^{(j)}\big) \right|^2 \\
				&
				+ \sum_{m \in \mathcal{M}_k} \left(
				\beta_{j,m} \, \mathrm{tr}\!\big(\boldsymbol{\Phi}_{k,m} \mathbf{R}_{j,m}\big) 
				- \left| \mathrm{tr}\!\big(\boldsymbol{\Phi}_{k,m}^{(j)}\big) \right|^2
				\right)
				\Bigg].
			\end{aligned}
			$}
		\label{eq:CI}
	\end{align}
	
	Finally, the effective noise power after combining is given by
	\begin{equation}
		\mathrm{N}_k = \sigma^2 \sum_{m \in \mathcal{M}_k} \mathrm{tr}(\boldsymbol{\Phi}_{k,m}).
		\label{eq:N}
	\end{equation}
	
	For URLLC applications requiring short packets with stringent reliability, the Shannon capacity is an optimistic benchmark. We adopt the FBL achievable rate approximation, which accounts for the rate penalty due to limited coding blocklength. Let $L_k$ denote the packet size in bits and $\epsilon_k^{\mathrm{dec}}$ the target decoding error probability. The achievable rate for UE $k$ is approximately
	\begin{equation}
		\resizebox{.95\hsize}{!}{$R_k \approx B \cdot \frac{\tau_c - \tau_p}{\tau_c} \left[ \log_2(1 + \mathrm{SINR}_k) - \sqrt{\frac{V_k}{L_k}} \frac{Q^{-1}(\epsilon_k^{\mathrm{dec}})}{\ln 2} \right]^+$},
		\label{eq:fbl_rate}
	\end{equation}
	where $B$ is the system bandwidth, $Q^{-1}(\cdot)$ is the inverse Gaussian Q-function, $[\cdot]^+ = \max\{0, \cdot\}$, and $V_k \approx 1 - (1 + \mathrm{SINR}_k)^{-2}$ is the channel dispersion \cite{10114604}. 
	
	\vspace{-3mm}
	\subsection{URLLC Traffic Model and QoS Requirements}
	\label{subsec:urllc_traffic}
	To capture the bursty traffic behavior typical of URLLC services, we model the packet arrivals of UE~$k$ using a two-state MMPP~\cite{fischer1993markov,latouche1999matrix}. The OFF and ON states are ordered as $(0,1)$, with ON-state arrival rate $\lambda_{\mathrm{ON},k}$, ON$\to$OFF transition rate $\xi_k$, and OFF$\to$ON transition rate $\nu_k$. The underlying phase generator and arrival-rate matrices are
	\begin{equation}
		\mathbf{Q}_k^{\mathrm{M}} =
		\begin{bmatrix}
			-\nu_k & \nu_k\\
			\xi_k & -\xi_k
		\end{bmatrix},
		\qquad
		\mathbf{\Lambda}_k =
		\begin{bmatrix}
			0 & 0\\
			0 & \lambda_{\mathrm{ON},k}
		\end{bmatrix},
		\label{eq:mmpp_generator}
	\end{equation}
	so that the average arrival rate is $\lambda_{\mathrm{avg},k}=\nu_k\lambda_{\mathrm{ON},k}/(\xi_k+\nu_k)$.
	
	Each UE is modeled as an MMPP/M/1 queue with exponential service time $S_k\sim\mathrm{Exp}(\mu_k)$ and service rate $\mu_k=R_k/L_k$, where $R_k$ is the achievable rate in~\eqref{eq:fbl_rate} and $L_k$ is the packet size in bits. Queue stability requires $\lambda_{\mathrm{avg},k}<\mu_k$. Let $\widetilde{S}_k(s)=\mu_k/(\mu_k+s)$ denote the service-time Laplace--Stieltjes transform. The associated queueing matrix $\mathbf{Z}_k$ is the unique substochastic solution of
	\begin{equation}
		\mathbf{Z}_k
		=
		\mu_k\Big(\mu_k\mathbf{I}_2-\mathbf{Q}_k^{\mathrm{M}}
		+\mathbf{\Lambda}_k-\mathbf{\Lambda}_k\mathbf{Z}_k\Big)^{-1},
		\label{eq:queue_matrix}
	\end{equation}
	and its invariant probability vector $\mathbf{g}_k$ satisfies $\mathbf{g}_k\mathbf{Z}_k=\mathbf{g}_k$ and $\mathbf{g}_k\mathbf{1}_2=1$. Then the stationary Laplace--Stieltjes transform of the queueing delay $W_k^{\mathrm{q}}$ is~\cite{lucantoni1991new}
	\begin{equation}
		\scalebox{0.95}{$
			\begin{aligned}
				\widetilde{W}_k^{\,\mathrm{q}}(s)
				= \frac{s\left(1-\lambda_{\mathrm{avg},k}/\mu_k\right)}
				{\lambda_{\mathrm{avg},k}}\,
				\mathbf{g}_k &
				\Big[s\mathbf{I}_2+\mathbf{Q}_k^{\mathrm{M}}
				-\mathbf{\Lambda}_k
				+\widetilde{S}_k(s)\mathbf{\Lambda}_k\Big]^{-1} \\
				& \mathbf{\Lambda}_k\mathbf{1}_2, \ \Re(s)\ge 0.
			\end{aligned}
			$}
		\label{eq:mmpp_waiting_lst}
	\end{equation}
		
	The end-to-end delay is $W_k=W_k^{\mathrm{q}}+S_k$, and the URLLC QoS requirement imposes
	\begin{equation}
		\Pr(W_k > D_{\max,k}) \le \epsilon_k^{\mathrm{delay}},
		\label{eq:delay_constraint}
	\end{equation}
	where $D_{\max,k}$ is the maximum tolerable latency and $\epsilon_k^{\mathrm{delay}}$ is the delay violation probability.\footnote{The delay violation probability $\epsilon_k^{\mathrm{delay}}$ is a queueing-level reliability metric and should not be confused with the decoding error probability $\epsilon_k^{\mathrm{dec}}$ in the finite-blocklength rate expression.} Since \eqref{eq:delay_constraint} does not admit a closed-form inversion for $\mu_k$, the minimum service rate
	$\mu_k^{\min}=\inf\{\mu_k>\lambda_{\mathrm{avg},k}:\Pr(W_k>D_{\max,k})\le\epsilon_k^{\mathrm{delay}}\}$ is obtained numerically from \eqref{eq:mmpp_waiting_lst}. The corresponding minimum rate requirement is $R_k^{\min}=\mu_k^{\min}L_k$, which, through \eqref{eq:fbl_rate}, yields the SINR constraint $\mathrm{SINR}_k \ge \bar{\gamma}_k$, where $\bar{\gamma}_k$ is obtained by inverting the finite-blocklength rate expression at $R_k=R_k^{\min}$.
	
	\vspace{-5mm}
	\section{Problem Formulation}
	\label{sec:problem_formulation}
	
	Before formulating the joint optimization problem of FA positions and pilot power allocation, we first define the channel estimation quality parameter and establish its relationship with system performance. 
	\vspace{-3mm}
	\subsection{Channel Estimation Quality Metric}
	\label{subsec:nmse_metric}
	
	We quantify channel estimation quality using the NMSE. For the link between UE $k$ and AP $m$, the NMSE is defined as
	\begin{equation}
		\mathrm{NMSE}_{k,m}(\mathbf{x}_m) = \frac{\mathbb{E}\{\|\mathbf{h}_{k,m} - \widehat{\mathbf{h}}_{k,m}\|^2\}}{\mathbb{E}\{\|\mathbf{h}_{k,m}\|^2\}}.
		\label{eq:nmse_def}
	\end{equation}
	Substituting the error covariance from \eqref{eq:error_covariance} yields
	\begin{equation}
		\mathrm{NMSE}_{k,m}(\mathbf{x}_m) = 1 - \frac{\mathrm{tr}(\boldsymbol{\Phi}_{k,m}(\mathbf{x}_m))}{\beta_{k,m} \, \mathrm{tr}(\mathbf{R}_{k,m}(\mathbf{x}_m))},
		\label{eq:nmse_explicit}
	\end{equation}
	where the dependence on antenna positions enters through both the channel estimate covariance $\boldsymbol{\Phi}_{k,m}$ and the spatial correlation matrix $\mathbf{R}_{k,m}$. This position dependence is a key enabler for fluid antenna optimization: by adjusting element locations, the system can reshape the spatial correlation structure to enhance channel estimation accuracy, thereby improving the resulting SINR.
	
	For UE $k$ served by multiple APs, we define the effective NMSE as the weighted average across the serving cluster
	\begin{equation}
		\mathrm{NMSE}_k(\mathbf{x}) = \frac{\sum_{m \in \mathcal{M}_k} \mathrm{tr}(\mathbf{C}_{k,m}(\mathbf{x}_m))}{\sum_{m \in \mathcal{M}_k} \beta_{k,m} \, \mathrm{tr}(\mathbf{R}_{k,m}(\mathbf{x}_m))},
		\label{eq:nmse_effective}
	\end{equation}
	where the weighting by $\beta_{k,m}\mathrm{tr}(\mathbf{R}_{k,m})$ emphasizes
	estimation quality at stronger links that dominate the combined
	signal.
	
	\vspace{-3mm}
	\subsection{SINR-NMSE Relationship via Eigenstructure Approximation}
	\label{subsec:sinr_nmse}
	
	The SINR in~\eqref{eq:sinr} involves trace products
	$\mathrm{tr}(\boldsymbol{\Phi}_{k,m}\mathbf{R}_{j,m})$ that depend on eigenstructure
	alignment between covariance matrices, precluding closed-form
	NMSE dependence~\cite{bjornson2017massive}. To obtain a tractable SINR-NMSE relationship, we introduce alignment factors that capture the geometric relationship between covariance matrices and establish an approximation with quantifiable error.
	
	\begin{definition}
		\label{def:alignment_factors}
		For UE $k$ and AP $m$, the self-alignment factor and cross-alignment factor are defined as
		\begin{align}
			\bar{\eta}_{k,m}(\mathbf{x}_m) &\triangleq \frac{\mathrm{tr}(\mathbf{R}_{k,m}^2(\mathbf{x}_m))}{\mathrm{tr}(\mathbf{R}_{k,m}(\mathbf{x}_m))}, \label{eq:eta_self} \\
			\bar{\eta}_{j,m}^{(k)}(\mathbf{x}_m) &\triangleq \frac{\mathrm{tr}(\mathbf{R}_{k,m}(\mathbf{x}_m) \mathbf{R}_{j,m}(\mathbf{x}_m))}{\mathrm{tr}(\mathbf{R}_{k,m}(\mathbf{x}_m))}. \label{eq:eta_cross}
		\end{align}
	\end{definition}
	
	The self-alignment factor $\bar{\eta}_{k,m}$ captures the eigenvalue concentration of $\mathbf{R}_{k,m}$, whereas $\bar{\eta}_{j,m}^{(k)}$ quantifies the spatial overlap between the covariance structures of UEs $k$ and $j$. Both factors are computable from statistical CSI and depend on antenna positions through the inter-element spacings in \eqref{eq:correlation_model}.
	
	\begin{lemma}
		\label{lem:trace_approx}
		Let $\rho_{k,m} = \tau_p p_k^{\mathrm{p}} \beta_{k,m} / \sigma^2$ denote the effective pilot SNR and $\kappa_{k,m} = \lambda_{\max}(\mathbf{R}_{k,m}) / \lambda_{\min}(\mathbf{R}_{k,m})$ the condition number of the spatial correlation matrix. For the no-pilot-contamination case ($\mathcal{P}_k = \{k\}$), the trace products satisfy
		\begin{align}
			\mathrm{tr}(\boldsymbol{\Phi}_{k,m} \mathbf{R}_{k,m}) &= \bar{\eta}_{k,m} \, \mathrm{tr}(\boldsymbol{\Phi}_{k,m}) + \mathcal{O}\left( \frac{(\kappa_{k,m} - 1)^2}{\rho_{k,m}} \right), \label{eq:trace_approx_self} \\
			\mathrm{tr}(\boldsymbol{\Phi}_{k,m} \mathbf{R}_{j,m}) &= \bar{\eta}_{j,m}^{(k)} \, \mathrm{tr}(\boldsymbol{\Phi}_{k,m}) + \mathcal{O}\left( \frac{\kappa_{k,m} - 1}{\rho_{k,m}} \right). \label{eq:trace_approx_cross}
		\end{align}
		The approximation error vanishes as $\rho_{k,m} \to \infty$ or $\kappa_{k,m} \to 1$.
	\end{lemma}
	
	\begin{IEEEproof}
		See Appendix~\ref{app:trace_approx}.
	\end{IEEEproof}
	
	The approximation is accurate when $\rho_{k,m}\gtrsim 10$~dB and
	$\kappa_{k,m}\lesssim 10$, conditions that are typically met in
	practical CF deployments with moderate angular spread ($\sigma_\theta\gtrsim 5^\circ$), as validated
	in Section~\ref{subsec:validation}. Under pilot contamination
	($|\mathcal{P}_k|>1$), $\rho_{k,m}$ is replaced by
	\begin{equation}
		\rho_{k,m}^{\mathrm{eff}} = \frac{\tau_p p_k^{\mathrm{p}} \beta_{k,m}}{\sum_{j \in \mathcal{P}_k \setminus \{k\}} \tau_p p_j^{\mathrm{p}} \beta_{j,m} + \sigma^2},
		\label{eq:rho_eff}
	\end{equation}
	and the bounds remain valid provided $\rho_{k,m}^{\mathrm{eff}}\gtrsim 10$~dB.
	
	Combining \eqref{eq:trace_approx_self}--\eqref{eq:trace_approx_cross} with \eqref{eq:nmse_explicit}, we express the estimate covariance trace as $\mathrm{tr}(\boldsymbol{\Phi}_{k,m}) = \beta_{k,m} \mathrm{tr}(\mathbf{R}_{k,m})(1 - \mathrm{NMSE}_{k,m})$. Substituting into the SINR components \eqref{eq:DS}--\eqref{eq:N} yields expressions where each term factors into position-dependent coefficients and NMSE-dependent scaling. To obtain a compact SINR-NMSE relationship suitable for deriving explicit QoS constraints, we introduce the following assumption.
	\vspace{-1mm}
	\begin{assumption}
		\label{ass:uniform_nmse}
		For each UE $k$, the per-AP NMSE values are approximately uniform across serving APs, i.e., $\mathrm{NMSE}_{k,m} \approx \overline{\mathrm{NMSE}}_k$ for all $m \in \mathcal{M}_k$.
	\end{assumption}
	
	{Assumption~\ref{ass:uniform_nmse} is introduced only in the reduced SINR mapping. The channel model, LMMSE covariance calculation, and per-AP NMSE evaluation retain the link-dependent large-scale fading and spatial correlation associated with every serving AP.} The following proposition characterizes when this assumption is accurate.
	\vspace{-1mm}
	\begin{proposition}
		\label{prop:nmse_variation}
		Consider a UC cluster $\mathcal{M}_k$ with large-scale fading coefficients satisfying $\beta_{k,m}\in[\beta_k^{\min},\beta_k^{\max}]$ for all $m\in\mathcal{M}_k$. Under uniform pilot power $p_k^{\mathrm p}$ and approximately identical spatial correlation matrices $\mathbf{R}_{k,m}\approx \mathbf{R}_k$ across serving APs, the NMSE variation satisfies
		\begin{equation}
			\max_{m,m'\in\mathcal{M}_k}
			\left|
			\mathrm{NMSE}_{k,m}-\mathrm{NMSE}_{k,m'}
			\right|
			\le \delta_k,
			\label{eq:nmse_variation_bound}
		\end{equation}
		where
		\begin{equation}
			\delta_k
			=
			\frac{\beta_k^{\max}-\beta_k^{\min}}{\beta_k^{\max}}
			\cdot
			\frac{\sigma^2}{\tau_p p_k^{\mathrm p}\beta_k^{\min}+\sigma^2}.
			\label{eq:delta_k}
		\end{equation}
	\end{proposition}
	\begin{IEEEproof}
		See Appendix~\ref{app:nmse_variation}.
	\end{IEEEproof}
	
	Proposition~\ref{prop:nmse_variation} shows that $\delta_k$ is controlled by  the relative large-scale fading
	spread within the cluster (limited by UC selection) and the
	noise-to-signal ratio at the weakest AP (diminishing at moderate
	pilot SNR). The sensitivity analysis in Section~\ref{subsec:validation} reports median absolute SINR deviations of only $0.017$ dB, supporting the use of Assumption~\ref{ass:uniform_nmse} over the dominant practical regime.
	
	
	Under Assumption~\ref{ass:uniform_nmse}, defining
	$\alpha_k = 1-\mathrm{NMSE}_k$, the aggregate coefficients are
	\begin{align}
		\Gamma_k &= \sum_{m \in \mathcal{M}_k} \beta_{k,m} \, \mathrm{tr}(\mathbf{R}_{k,m}), \label{eq:Gamma} \\
		\Omega_k &= \sum_{m \in \mathcal{M}_k} \bar{\eta}_{k,m} \beta_{k,m}^2 \, \mathrm{tr}(\mathbf{R}_{k,m}), \label{eq:Omega} \\
		\Xi_{k,j} &= \sum_{m \in \mathcal{M}_k} \bar{\eta}_{j,m}^{(k)} \beta_{j,m} \beta_{k,m} \, \mathrm{tr}(\mathbf{R}_{k,m}), \label{eq:Xi}
	\end{align}
	and the SINR admits the reduced quadratic-fractional surrogate
	\begin{equation}
		\mathrm{SINR}_k \approx \frac{A_k \alpha_k^2}{B_k \alpha_k + C_k},
		\label{eq:sinr_quadratic}
	\end{equation}
	where $A_k = p_k^{\mathrm{u}}\Gamma_k^2$,\;
	$B_k = p_k^{\mathrm{u}}\Omega_k + \sum_{j\notin\mathcal{P}_k}p_j^{\mathrm{u}}\Xi_{k,j}+\sigma^2\Gamma_k$,\;
	$C_k = \mathrm{CI}_k$. The coherent interference $\mathrm{CI}_k$ in \eqref{eq:CI} depends implicitly on $\alpha_k$ through the estimate covariances $\Phi_{k,m}$ and $\Phi_{k,m}^{(j)}$. In the surrogate \eqref{eq:sinr_quadratic}, $C_k$ is evaluated pointwise from the current antenna positions and pilot powers, yielding a scalar coefficient used for local SINR prediction and threshold inversion. These coefficients are recomputed whenever the thresholds are updated (see Section~\ref{sec:proposed_solution}). 
	
	\vspace{-3mm}
	\subsection{NMSE Threshold for QoS Guarantee}
	\label{subsec:nmse_threshold}
	
	The quadratic SINR-NMSE relationship enables explicit characterization of the estimation quality required to meet URLLC QoS constraints.
	
	\begin{theorem}
		\label{thm:nmse_threshold}
		Under Assumption~\ref{ass:uniform_nmse}, the SINR constraint $\mathrm{SINR}_k \geq \bar{\gamma}_k$ is satisfied if and only if $\overline{\mathrm{NMSE}}_k \leq \varepsilon_k$, where
		\begin{equation}
			\varepsilon_k = 1 - \frac{\bar{\gamma}_k B_k + \sqrt{(\bar{\gamma}_k B_k)^2 + 4 A_k \bar{\gamma}_k C_k}}{2 A_k}.
			\label{eq:nmse_threshold}
		\end{equation}
	\end{theorem}
	
	\begin{IEEEproof}
		Substituting $\alpha_k = 1 - \overline{\mathrm{NMSE}}_k$ into the SINR constraint $A_k \alpha_k^2 / (B_k \alpha_k + C_k) \geq \bar{\gamma}_k$ and rearranging yields the quadratic inequality $A_k \alpha_k^2 - \bar{\gamma}_k B_k \alpha_k - \bar{\gamma}_k C_k \geq 0$. Since $A_k > 0$, the parabola opens upward. Computing the roots and noting that the negative root is infeasible (as $\alpha_k > 0$), the constraint is satisfied for $\alpha_k \geq \alpha_+$ where $\alpha_+$ is the positive root. Converting back to NMSE gives \eqref{eq:nmse_threshold}.
	\end{IEEEproof}
	
	The threshold $\varepsilon_k$ depends on positions and powers through
	$A_k$, $B_k$, and $C_k$, so FA positioning and pilot power control can
	relax the estimation requirement. For orthogonal pilots
	($\mathcal{P}_k=\{k\}$, $C_k=0$), \eqref{eq:nmse_threshold} reduces to
	$\varepsilon_k = 1-\bar{\gamma}_k B_k/A_k$, directly linking stricter SINR
	targets to lower NMSE requirements.
	
	\vspace{-3mm}
	\subsection{Joint Optimization Problem}
	\label{subsec:optimization_problem}
	
	We formulate the joint optimization of fluid antenna positions and pilot power allocation to minimize average NMSE while guaranteeing URLLC QoS for all UEs. The optimization exploits statistical CSI, comprising large-scale fading coefficients $\{\beta_{k,m}\}$ and spatial correlation matrices $\{\mathbf{R}_{k,m}\}$, which vary on a slow timescale compared to small-scale fading. The joint optimization problem is formulated as follows: 
	
	\begin{subequations}
		\label{eq:opt_problem}
		\begin{align}
			\mathcal{P}: \min_{\mathbf{x}, \mathbf{p}^{\mathrm{p}}} \quad & \frac{1}{K} \sum_{k=1}^{K} \mathrm{NMSE}_k(\mathbf{x}, \mathbf{p}^{\mathrm{p}}) \label{eq:objective} \\
			\text{s.t.} \quad & \mathrm{NMSE}_k \leq \varepsilon_k(\mathbf{x},  \mathbf{p}^{\mathrm{p}}), \; \forall k, \label{eq:qos_constraint} \\
			& 0 \leq x_m(n) \leq W_m, \; \forall m, n, \label{eq:boundary} \\
			& x_m(n+1) - x_m(n) \geq d_{\min}, \; \forall m, n, \label{eq:spacing} \\
			& P_{\min}^{\mathrm{p}} \leq p_k^{\mathrm{p}} \leq P_{\max}^{\mathrm{p}}, \; \forall k. \label{eq:power_limits}
		\end{align}
	\end{subequations}
	
	The objective \eqref{eq:objective} minimizes the average effective NMSE across all UEs, which improves channel-estimation quality and, through the reduced SINR--NMSE relation in~\eqref{eq:sinr_quadratic}, tends to increase the effective SINR. Constraint~\eqref{eq:qos_constraint} enforces the NMSE threshold derived in  Theorem~\ref{thm:nmse_threshold}, thereby guaranteeing the target SINR required for URLLC delay satisfaction. Specifically, the achieved NMSE depends on the estimation covariance
	through $\boldsymbol{\Phi}_{k,m}$ and $\mathbf{R}_{k,m}$, while the
	threshold $\varepsilon_k$ depends on the reduced SINR coefficients
	$A_k$, $B_k$, and $C_k$. Constraints \eqref{eq:boundary} and \eqref{eq:spacing} enforce physical realizability with minimum spacing $d_{\min}$ and implicit element
	ordering, while constraint \eqref{eq:power_limits} limits pilot power within $[P_{\min}^{\mathrm{p}}, P_{\max}^{\mathrm{p}}]$ according to hardware and interference considerations.
	
	Problem~$\mathcal{P}$ is nonconvex because of the matrix inverse $\boldsymbol{\Psi}_{k,m}^{-1}(\mathbf{x}_m)$ in \eqref{eq:phi_matrix}, the transcendental dependence induced by the position-dependent correlation model in \eqref{eq:correlation_model}, and the implicit coupling in \eqref{eq:qos_constraint} between the
	achieved NMSE and the variable-dependent threshold. Similar coupled
	resource-control formulations are known to reduce to hard
	sum-of-ratios programs~\cite{freund2001}, motivating the AO framework developed next.
	
	\vspace{-3mm}
	\section{Proposed Solution}
	\label{sec:proposed_solution}
	
	Given the nonconvexity of Problem~$\mathcal{P}$, we adopt an AO approach that 
	cycles between antenna position optimization and pilot power 
	allocation. Since each subproblem reduces the shared objective 
	over its own variables, the overall objective is non-increasing 
	across iterations, guaranteeing convergence to a stationary point.
	
	\vspace{-4mm}
	\subsection{Antenna Position Optimization}
	\label{subsec:position_opt}
	
	Given fixed pilot powers $\mathbf{p}^{\mathrm{p}}$, the antenna position optimization subproblem is
	\begin{equation}
		\mathcal{P}_{\mathbf{x}}: \min_{\mathbf{x}} \; \frac{1}{K} \sum_{k=1}^{K} \mathrm{NMSE}_k(\mathbf{x})
		\quad \text{s.t.} \; \eqref{eq:qos_constraint}, \eqref{eq:boundary}, \eqref{eq:spacing},
		\label{eq:subproblem_x}
	\end{equation}
	where the dependence of both the achieved NMSE and threshold $\varepsilon_k$ on antenna positions enters through \eqref{eq:nmse_effective} and \eqref{eq:nmse_threshold}, respectively. Subproblem $\mathcal{P}_{\mathbf{x}}$ is a continuous nonconvex optimization problem because both the NMSE and the threshold depend nonlinearly on the position-dependent covariance matrices through \eqref{eq:phi_matrix} and \eqref{eq:correlation_model}. Additionally, the QoS constraint \eqref{eq:qos_constraint} couples the achieved NMSE with the position-dependent threshold, creating an implicit constraint structure.
	
	Despite nonconvexity, constraints \eqref{eq:boundary} and \eqref{eq:spacing} jointly define the convex polytope
	\begin{equation}
		\mathcal{X}_m = \left\{ \mathbf{x}_m \in \mathbb{R}^N:\mathbf{0} \leq \mathbf{x}_m \leq W_m \mathbf{1}, \;  \mathbf{D} \mathbf{x}_m \geq d_{\min} \mathbf{1} \right\},
		\label{eq:feasible_set}
	\end{equation}
	where $\mathbf{D} \in \mathbb{R}^{(N-1) \times N}$ is the first-order differencing matrix with entries $[\mathbf{D}]_{i,i} = -1$ and $[\mathbf{D}]_{i,i+1} = 1$. The convexity of $\mathcal{X}_m$ admits efficient Euclidean projection. We exploit this structure via SCA, which iteratively constructs strongly convex local surrogates over the feasible polytope while updating the QoS thresholds according to the current antenna configuration. The QoS constraint is incorporated through the penalized objective
	\begin{equation}
		\begin{aligned}
			\mathcal{L}(\mathbf{x}; \mu)
			&= \frac{1}{K} \sum_{k=1}^{K} \mathrm{NMSE}_k(\mathbf{x}) \\
			&\quad + \mu \sum_{k=1}^{K}
			\left[\max\!\left(0,\ \mathrm{NMSE}_k(\mathbf{x}) - \varepsilon_k(\mathbf{x})\right)\right]^2,
		\end{aligned}
		\label{eq:augmented_lagrangian}
	\end{equation}
	where $\mu > 0$ is the penalty parameter. The quadratic penalty progressively enforces constraint satisfaction as $\mu$ increases while preserving a smooth dependence on the violation magnitude away from the hinge point. To maintain tractability within SCA, the QoS threshold is evaluated at the current iterate and held fixed during the local first-order step. Specifically, at iteration $t$, we define
	$\varepsilon_k^{(t)} \triangleq \varepsilon_k(\mathbf{x}^{(t)})$ for all
	$k$, and denote by $\mathcal{L}^{(t)}(\mathbf{x};\mu)$ the local
	frozen-threshold objective obtained from $\mathcal{L}(\mathbf{x};\mu)$
	by replacing $\varepsilon_k(\mathbf{x})$ with $\varepsilon_k^{(t)}$. {The reduced SINR coefficients $\{A_k,B_k,C_k\}$ in 
		\eqref{eq:sinr_quadratic} are recomputed together with 
		$\varepsilon_k^{(t)}$ using the current position iterate 
		$\mathbf{x}^{(t)}$ and the fixed pilot power vector, and are then 
		held fixed during the subsequent frozen-threshold SCA step.}  SCA then linearizes $\mathcal{L}^{(t)}$ around the current iterate $\mathbf{x}^{(t)}$ and adds proximal regularization:
	\begin{align}
		\tilde{\mathcal{L}}^{(t)}(\mathbf{x}; \mathbf{x}^{(t)})
		&=
		\mathcal{L}^{(t)}(\mathbf{x}^{(t)};\mu)
		+
		\nabla \mathcal{L}^{(t)}(\mathbf{x}^{(t)};\mu)^\mathsf{T}
		(\mathbf{x} - \mathbf{x}^{(t)}) \nonumber \\
		&+
		\frac{\rho}{2}\| \mathbf{x} - \mathbf{x}^{(t)} \|^2,
		\label{eq:quadratic_surrogate}
	\end{align}
	where $\rho > 0$ governs the trust-region radius. The proximal term ensures strong convexity of the local model and moderates the step size to maintain stable descent.  Minimizing \eqref{eq:quadratic_surrogate} subject to $\mathbf{x}_m \in \mathcal{X}_m$ admits the closed-form projected update
	\begin{equation}
		\mathbf{x}^{(t+1)}_m
		=
		\mathrm{proj}_{\mathcal{X}_m}
		\left(
		\mathbf{x}^{(t)}_m
		-
		\rho^{-1}\nabla_{\mathbf{x}_m}\mathcal{L}^{(t)}(\mathbf{x}^{(t)};\mu)
		\right),
		\ \forall m,
		\label{eq:sca_update}
	\end{equation}
	where $\mathrm{proj}_{\mathcal{X}_m}(\cdot)$ denotes Euclidean projection onto $\mathcal{X}_m$. The search direction is therefore computed from the local frozen-threshold objective $\mathcal{L}^{(t)}$, whereas the candidate point is accepted only if the true penalized objective $\mathcal{L}(\cdot;\mu)$ decreases.
	
	The gradient $\nabla_{\mathbf{x}_m}\mathcal{L}^{(t)}$ aggregates contributions from both the NMSE objective and the local penalty term. Defining the current QoS violation under the frozen threshold as $g_k^{(t)}
	\triangleq
	\max\!\left(
	0,\ \mathrm{NMSE}_k(\mathbf{x}^{(t)}) - \varepsilon_k^{(t)}
	\right)$, the partial derivative with respect to position $x_m(n)$ is
	\begin{equation}
		\frac{\partial \mathcal{L}^{(t)}}{\partial x_m(n)}
		=
		\frac{1}{K}
		\sum_{k \in \mathcal{K}_m}
		\left( 1 + 2K\mu g_k^{(t)} \right)
		\frac{\partial \mathrm{NMSE}_k}{\partial x_m(n)},
		\label{eq:gradient_expression}
	\end{equation}
	where $\mathcal{K}_m$ denotes the UEs served by AP $m$. The factor $\left(1 + 2K\mu g_k^{(t)}\right)$ increases the gradient weight of UEs with larger QoS violations, so infeasible UEs exert stronger influence on the position update, thereby steering the iterates toward feasibility. The NMSE gradient is obtained by differentiating \eqref{eq:nmse_effective}:
	\begin{align}
		\frac{\partial \mathrm{NMSE}_k}{\partial x_m(n)}
		&= -\frac{1}{\beta_{k,m}\,\mathrm{tr}\!\left(\mathbf{R}_{k,m}\right)}
		\Biggl(
		\frac{\partial\,\mathrm{tr}\!\left(\boldsymbol{\Phi}_{k,m}\right)}{\partial x_m(n)}
		\nonumber\\
		&\qquad\quad
		- \mathrm{NMSE}_{k,m}\,\beta_{k,m}\,
		\frac{\partial\,\mathrm{tr}\!\left(\mathbf{R}_{k,m}\right)}{\partial x_m(n)}
		\Biggr),
		\label{eq:nmse_gradient}
	\end{align}
	where the correlation matrix gradient $\partial \mathbf{R}_{k,m} / \partial x_m(n)$ follows from \eqref{eq:correlation_model}.
	
	\begin{algorithm}[t]
		\small
		\caption{SCA for Antenna Position Optimization}
		\label{alg:sca_position}
		\SetAlgoNoEnd
		\KwIn{Initial $\mathbf{x}^{(0)}$, $\rho_0$, $\mu_0$, $\epsilon$, $I_{\max}=100$, $\delta=10^{-3}$}
		\KwOut{Optimized positions $\mathbf{x}^\star$}
		\textbf{Initialize:} $\mathbf{x}^{(0)}_m$ at feasible region centroid, $\forall m$ \nllabel{line:init}\;
		Set $\rho \leftarrow \rho_0$, $\mu \leftarrow \mu_0$ \nllabel{line:param}\;
		\For{$t = 0, 1, \ldots, I_{\max}-1$}{ \nllabel{line:loop}
			Compute $\mathbf{R}_{k,m}(\mathbf{x}^{(t)})$, $\boldsymbol{\Phi}_{k,m}(\mathbf{x}^{(t)})$, $\forall k,m$ \nllabel{line:cov}\;
			{Evaluate $\mathrm{NMSE}_k(\mathbf{x}^{(t)})$, set 
				$\varepsilon_k^{(t)}=\varepsilon_k(\mathbf{x}^{(t)})$, and recompute 
				the reduced SINR coefficients $\{A_k,B_k,C_k\}$ using 
				$\mathbf{x}^{(t)}$ and the fixed pilot powers, $\forall k$} 
			\nllabel{line:nmse}\;
			Compute $\nabla_{\mathbf{x}_m}\mathcal{L}^{(t)}(\mathbf{x}^{(t)};\mu)$ via \eqref{eq:gradient_expression}--\eqref{eq:nmse_gradient}, treating $\varepsilon_k^{(t)}$ as fixed, $\forall m$ \nllabel{line:grad}\;
			$\mathbf{x}^{(t+1)}_m \leftarrow \mathrm{proj}_{\mathcal{X}_m}\!\left( \mathbf{x}^{(t)}_m - \rho^{-1} \nabla_{\mathbf{x}_m}\mathcal{L}^{(t)} \right)$, $\forall m$ \nllabel{line:update}\;
			Evaluate $\mathcal{L}(\mathbf{x}^{(t+1)};\mu)$ via \eqref{eq:augmented_lagrangian} \nllabel{line:evalnext}\;
			\lIf{$\mathcal{L}(\mathbf{x}^{(t+1)};\mu) < \mathcal{L}(\mathbf{x}^{(t)};\mu)$}{$\rho \leftarrow \max(\rho_{\min}, \alpha_{\mathrm{dec}}\rho)$} \nllabel{line:accept}
			\lElse{$\rho \leftarrow \min(\rho_{\max}, \alpha_{\mathrm{inc}}\rho)$; $\mathbf{x}^{(t+1)} \leftarrow \mathbf{x}^{(t)}$} \nllabel{line:reject}
			\If{$\max_k g_k^{(t)} > \delta$ \textbf{and} $t \bmod T_\mu = 0$}{ \nllabel{line:mu_check}
				$\mu \leftarrow \min(\mu_{\max}, c_{\mu} \mu)$ \nllabel{line:mu_update}\;
			}
			\lIf{$\| \mathbf{x}^{(t+1)} - \mathbf{x}^{(t)} \| < \epsilon$}{\textbf{break}} \nllabel{line:stop}
		}
		\Return $\mathbf{x}^\star \leftarrow \mathbf{x}^{(t+1)}$ \nllabel{line:return}
	\end{algorithm}
	
	Algorithm~\ref{alg:sca_position} summarizes the procedure. Lines~\ref{line:cov}--\ref{line:update} compute the projected SCA step using the local
	frozen-threshold objective $\mathcal{L}^{(t)}$, whereas Lines~\ref{line:evalnext}--\ref{line:reject} accept the candidate only if the true penalized objective
	$\mathcal{L}(\mathbf{x};\mu)$ in \eqref{eq:augmented_lagrangian} decreases. Projection onto $\mathcal{X}_m$ exploits isotonic regression structure for $\mathcal{O}(N)$ complexity\footnote{Positions satisfying
		$x_m(1)\leq\cdots\leq x_m(N)$ with minimum gaps reduce projection to
		the pool-adjacent-violators algorithm followed by boundary clipping,
		achieving $O(N)$ versus $O(N^3)$ for general quadratic programming.}. Lines~\ref{line:accept}--\ref{line:reject} moderate the local step size: successful steps reduce $\rho$ to accelerate progress, while unsuccessful steps increase $\rho$ and reject the candidate. The penalty parameter $\mu$ is periodically increased in Lines~\ref{line:mu_check}--\ref{line:mu_update} when QoS violations persist. Here, $\rho_{\min}$ and $\rho_{\max}$ bound the adaptive regularization, $\alpha_{\mathrm{dec}}$ and $\alpha_{\mathrm{inc}}$ are contraction/expansion factors, $T_\mu$ is the penalty-update period, and $c_\mu>1$ is the penalty-scaling factor.
	
	\vspace{-1mm}
	\begin{proposition}
		\label{prop:sca_convergence}
		Let $\{\mathbf{x}^{(t)}\}$ denote the sequence of accepted iterates
		generated by Algorithm~\ref{alg:sca_position}. For fixed~$\mu$, the
		sequence $\{\mathcal{L}(\mathbf{x}^{(t)};\mu)\}$ is monotonically
		non-increasing. Moreover, if $\{\mathbf{x}^{(t)}\}$ is asymptotically
		regular, every accumulation point $\mathbf{x}^\star$ is a stationary
		point of the limiting frozen-threshold objective with
		$\varepsilon_k=\varepsilon_k(\mathbf{x}^\star)$ for all~$k$.
	\end{proposition}
	\vspace{-2mm}
	\begin{proof}
		By Line~\ref{line:accept}, $\mathcal{L}(\mathbf{x}^{(t+1)};\mu)
		<\mathcal{L}(\mathbf{x}^{(t)};\mu)$ for every accepted iterate, so
		the sequence is non-increasing and, being bounded below on compact
		$\prod_m\mathcal{X}_m$, convergent. Along any subsequence
		$\mathbf{x}^{(t_j)}\!\to\!\mathbf{x}^\star$, asymptotic regularity
		gives $\|\mathbf{x}^{(t_j+1)}\!-\!\mathbf{x}^{(t_j)}\|\!\to\!0$,
		continuity of~\eqref{eq:nmse_threshold} yields
		$\varepsilon_k^{(t_j)}\!\to\!\varepsilon_k(\mathbf{x}^\star)$ for
		all~$k$, and $\rho\!\leq\!\rho_{\max}$ ensures the projected-gradient
		mapping of the local frozen-threshold model vanishes. Hence, $\mathbf{x}^\star$ satisfies the variational inequality of the limiting frozen-threshold objective and is therefore a stationary point of that objective.
	\end{proof}
	
	
	
	\subsubsection{Computational Complexity}
	\label{para:comput_SCA}
	The per-iteration cost is governed by gradient evaluation. Computing $\partial \mathbf{R}_{k,m} / \partial x_m(n)$ requires $\mathcal{O}(N^2)$ operations per UE-AP pair, while LMMSE covariance computation involves $\mathcal{O}(N^3)$ matrix inversion. Aggregating across $K$ UEs, $M$ APs, and $N$ elements yields per-iteration complexity $\mathcal{O}(KMN^3)$. With at most $I_{\max}$ iterations, total complexity becomes $\mathcal{O}(I_{\max} KMN^3)$, rendering Algorithm~\ref{alg:sca_position} suitable for large-scale CF deployments.
	
	
	\vspace{-4mm}
	\subsection{Pilot Power Allocation}\label{sec:bcd_pilot}
	
	With antenna positions $\{\mathbf{x}_m^\star\}_{m=1}^M$ fixed from Algorithm~\ref{alg:sca_position}, we optimize the pilot power vector $\mathbf{p}^{\mathrm{p}} = [p_1^{\mathrm{p}},\ldots,p_K^{\mathrm{p}}]^\mathsf{T}$ through a regularized formulation that balances estimation
	quality against energy consumption:
	\begin{subequations}\label{eq:pilot_power_sub}
		\begin{align}
			\min_{\mathbf{p}^{\mathrm{p}}} \quad
			& \frac{1}{K}\sum_{k=1}^{K} \mathrm{NMSE}_k(\mathbf{p}^{\mathrm{p}})
			+ \frac{\omega}{K P_{\max}^{\mathrm{p}}}
			\sum_{k=1}^{K} p_k^{\mathrm{p}}
			\label{eq:pilot_obj}\\
			\mathrm{s.t.} \quad
			& \eqref{eq:qos_constraint},\; \eqref{eq:power_limits},
		\end{align}
	\end{subequations}
	where $\omega > 0$ is a weighting factor. The power 
	regularization term in~\eqref{eq:pilot_obj} provides a systematic preference for lower pilot powers, avoiding
	convergence to full-power solutions when pilot contamination is weak and promoting energy-efficient operation in contamination-saturated regimes.
	
	To formalize the dependency on $\mathbf{p}^{\mathrm{p}}$, 
	we write $\mathrm{NMSE}_k(\mathbf{p}^{\mathrm{p}}) = 
	1 - \Upsilon_k(\mathbf{p}^{\mathrm{p}})/\Gamma_k$, where 
	$\Gamma_k$ is defined 
	in~\eqref{eq:Gamma}, and
	\begin{equation}\label{eq:Upsilon_k}
		\Upsilon_k(\mathbf{p}^{\mathrm{p}})
		= \sum_{m \in \mathcal{M}_k} \tau_p\, p_k^{\mathrm{p}}\, 
		\beta_{k,m}^2\,
		\mathrm{tr}\!\left(\mathbf{R}_{k,m}
		\boldsymbol{\Psi}_{k,m}^{-1}\mathbf{R}_{k,m}\right).
	\end{equation}
	The coupling arises from 
	$\boldsymbol{\Psi}_{k,m}$ in \eqref{eq:psi_matrix}, where $p_k^{\mathrm{p}}$ enters both the numerator scaling of~\eqref{eq:Upsilon_k} and the denominator inverse through $\boldsymbol{\Psi}_{k,m}$. This structure motivates the penalty-based BCD formulation developed next.
	
	To manage the nonconvexity and QoS constraints, we employ a penalty-based BCD framework. Let $\mathcal{L}(\mathbf{p}^{\mathrm{p}};\mu)$ denote the true penalized objective obtained by augmenting \eqref{eq:pilot_obj} with the squared hinge penalties $\mu\sum_{k=1}^{K}[\max(0,\mathrm{NMSE}_k(\mathbf{p}^{\mathrm{p}})-\varepsilon_k(\mathbf{p}^{\mathrm{p}}))]^2$, where $\mu>0$ is the penalty parameter. At outer iteration $t$, we update the QoS thresholds according to the current power vector, $\varepsilon_k^{(t)} \triangleq \varepsilon_k(\mathbf{p}^{\mathrm{p},(t)})$, $\forall k$, and define the corresponding frozen-threshold penalized objective
	\begin{equation}
		\begin{aligned}
			\mathcal{L}^{(t)}(\mathbf{p}^{\mathrm{p}};\mu)
			& \triangleq
			\frac{1}{K}\sum_{k=1}^{K}\mathrm{NMSE}_k(\mathbf{p}^{\mathrm{p}})
			+
			\frac{\omega}{K P_{\max}^{\mathrm{p}}}\sum_{k=1}^{K} p_k^{\mathrm{p}}\\
			&+
			\mu\sum_{k=1}^{K}
			\left[\max\!\left(0,\,
			\mathrm{NMSE}_k(\mathbf{p}^{\mathrm{p}})
			-\varepsilon_k^{(t)}\right)\right]^2.
		\end{aligned}
	\end{equation}
	
	{
		Along with the thresholds, the reduced SINR coefficients 
		$\{A_k,B_k,C_k\}$ in \eqref{eq:sinr_quadratic} are recomputed from 
		the current pilot power vector $\mathbf{p}^{\mathrm{p},(t)}$ and the 
		fixed antenna positions at the beginning of each outer BCD iteration, 
		and are held fixed during the subsequent inner BCD sweep.} To obtain analytically tractable scalar BCD updates, we further construct a violation-weighted surrogate by freezing the current penalty slopes. Defining $g_k^{(t)}
	\triangleq
	\max\!\left(
	0,\,
	\mathrm{NMSE}_k(\mathbf{p}^{\mathrm{p},(t)})
	-\varepsilon_k^{(t)}
	\right),
	\ \forall k$, the surrogate objective used within the inner BCD sweep is
	\begin{equation}
		\widehat{\mathcal{L}}^{(t)}(\mathbf{p}^{\mathrm{p}};\mu)
		\triangleq
		\sum_{k=1}^{K} w_k^{(t)}\,\mathrm{NMSE}_k(\mathbf{p}^{\mathrm{p}})
		+
		\frac{\omega}{K P_{\max}^{\mathrm{p}}}\sum_{k=1}^{K} p_k^{\mathrm{p}},
	\end{equation}
	where $w_k^{(t)} \triangleq \frac{1}{K} + 2\mu g_k^{(t)}, \ \forall k$.
	
	This construction mirrors the local frozen-threshold treatment used in the position-optimization stage: the thresholds and violation weights are updated at the beginning of each outer iteration and then held fixed during the subsequent BCD sweep.
	
	\subsubsection{BCD Decomposition and Derivatives}
	
	The BCD strategy leverages the pilot-contamination structure whereby modifying $p_l^{\mathrm{p}}$ affects the NMSE only of UEs in the co-pilot set $\mathcal{P}_l$. With the weights $\{w_k^{(t)}\}$ fixed over the current outer iteration, the scalar subproblem for UE~$l$ is obtained by restricting $\widehat{\mathcal{L}}^{(t)}$ to the terms that depend on $p_l^{\mathrm{p}}$, namely
	\begin{equation}\label{eq:bcd_sub}
		\min_{p_l^{\mathrm{p}} \in 
			[P_{\min}^{\mathrm{p}}, P_{\max}^{\mathrm{p}}]} \;
		f_l^{(t)}(p_l^{\mathrm{p}}) \triangleq
		\sum_{k \in \mathcal{P}_l} w_k^{(t)}\,
		\mathrm{NMSE}_k(p_l^{\mathrm{p}})
		+ \frac{\omega}{K P_{\max}^{\mathrm{p}}}\, 
		p_l^{\mathrm{p}}.
	\end{equation}
	
	Thus, UEs with larger current QoS violations receive greater weight during the inner BCD sweep, which biases the updates toward restoring feasibility while retaining a scalar closed-form derivative structure.
	
	To enable a second-order solution, we derive the exact 
	derivatives of~\eqref{eq:bcd_sub}. Let 
	$\mathbf{V}_{k,m} \triangleq 
	\boldsymbol{\Psi}_{k,m}^{-1}\mathbf{R}_{k,m}$ and 
	define the auxiliary trace quantities 
	$A_n^{(l,m)} \triangleq 
	\mathrm{tr}(\mathbf{R}_{l,m} \mathbf{V}_{l,m}^n)$ for 
	$n \in \{1,2,3\}$, 
	$B_1^{(k,l,m)} \triangleq 
	\|\mathbf{R}_{l,m}^{1/2}\mathbf{V}_{k,m}\|_F^2$, and 
	$B_2^{(k,l,m)} \triangleq 
	\|\boldsymbol{\Psi}_{k,m}^{-1/2}\mathbf{R}_{l,m}
	\mathbf{V}_{k,m}\|_F^2$.
	\vspace{-2mm}
	\begin{lemma}\label{lem:grad_hess}
		The first and second derivatives of 
		$f_l^{(t)}(p_l^{\mathrm{p}})$ are
		\begin{align}
			\bigl(f_l^{(t)}\bigr)'(p_l^{\mathrm{p}}) &= g_l^{\mathrm{self}} 
			+ g_l^{\mathrm{cross}} 
			+ \frac{\omega}{K P_{\max}^{\mathrm{p}}}, 
			\label{eq:gradient}\\
			\bigl(f_l^{(t)}\bigr)''(p_l^{\mathrm{p}}) &= h_l^{\mathrm{self}} 
			+ h_l^{\mathrm{cross}}, \label{eq:hessian}
		\end{align}
		where the components are
		\begin{subequations}
			\small
			\allowdisplaybreaks
			\begin{align}
				g_l^{\mathrm{self}} &= -\frac{w_l^{(t)}}{\Gamma_l}
				\sum_{m \in \mathcal{M}_l}
				\tau_p\,\beta_{l,m}^2
				\left[A_1^{(l,m)} - \tau_p\, p_l^{\mathrm{p}}\,
				\beta_{l,m}\, A_2^{(l,m)}\right], 
				\label{eq:g_self}\\
				g_l^{\mathrm{cross}} &= 
				\sum_{k \in \mathcal{P}_l \setminus \{l\}}
				\frac{w_k^{(t)}}{\Gamma_k}\sum_{m \in \mathcal{M}_k}
				\tau_p^2\, p_k^{\mathrm{p}}\, \beta_{k,m}^2\, 
				\beta_{l,m}\, B_1^{(k,l,m)}, 
				\label{eq:g_cross}\\
				h_l^{\mathrm{self}} &= \frac{w_l^{(t)}}{\Gamma_l}
				\sum_{m \in \mathcal{M}_l}
				2\tau_p^2\,\beta_{l,m}^3
				\left[A_2^{(l,m)} - \tau_p\, p_l^{\mathrm{p}}\, 
				\beta_{l,m}\, A_3^{(l,m)}\right], 
				\label{eq:h_self}\\
				h_l^{\mathrm{cross}} &= 
				-\sum_{k \in \mathcal{P}_l \setminus \{l\}}
				\frac{w_k^{(t)}}{\Gamma_k}\sum_{m \in \mathcal{M}_k}
				2\tau_p^3\, p_k^{\mathrm{p}}\, \beta_{k,m}^2\, 
				\beta_{l,m}^2\, B_2^{(k,l,m)}. 
				\label{eq:h_cross}
			\end{align}
		\end{subequations}
	\end{lemma}
	\vspace{-1mm}
	\begin{proof}
		Applying $d\boldsymbol{\Psi}_{k,m}^{-1}/dp_l^{\mathrm{p}}
		= -\boldsymbol{\Psi}_{k,m}^{-1}\mathbf{D}_{l,m}
		\boldsymbol{\Psi}_{k,m}^{-1}$ with
		$\mathbf{D}_{l,m}=\tau_p\beta_{l,m}\mathbf{R}_{l,m}$ and differentiating $\Upsilon_k$ via the product rule for the self-term ($k=l$) and the chain rule for cross-terms ($k\neq l$) yields~\eqref{eq:g_self}--\eqref{eq:g_cross}. A second differentiation introduces
		$2\boldsymbol{\Psi}_{k,m}^{-1}\mathbf{D}_{l,m}
		\boldsymbol{\Psi}_{k,m}^{-1}\mathbf{D}_{l,m}
		\boldsymbol{\Psi}_{k,m}^{-1}$, producing the higher-order traces in~\eqref{eq:h_self}--\eqref{eq:h_cross}. Since the weights $w_k^{(t)}$ enter $f_l^{(t)}$ only as scalar multiplicative coefficients, the derivative expressions retain the same functional form regardless of the specific weight definition.
	\end{proof}
	
	\vspace{-1mm}
	\subsubsection{Safeguarded Newton Solver}
	
	The Hessian component $h_l^{\mathrm{self}}$ 
	in~\eqref{eq:h_self} can change sign depending on the 
	operating regime, rendering the subproblem non-globally 
	convex. We therefore employ a safeguarded Newton method. 
	At inner iteration $j$, the power is updated as
	\begin{equation}\label{eq:newton_update}
		p_l^{(j+1)} = \Pi_{[P_{\min}^{\mathrm{p}}, 
			P_{\max}^{\mathrm{p}}]}\!\left(
		p_l^{(j)} - \alpha^{(j)} d^{(j)}\right),
	\end{equation}
	where $d^{(j)}=(f_l^{(t)})'/(f_l^{(t)})''$ if
	$(f_l^{(t)})''(p_l^{(j)})>\delta_H=10^{-10}$, and
	$d^{(j)}=\mathrm{sgn}((f_l^{(t)})')\cdot 0.1(P_{\max}^{\mathrm{p}}-P_{\min}^{\mathrm{p}})$
	otherwise. 
	The step size $\alpha^{(j)}$ is initialized to one for 
	Newton steps and determined via Armijo backtracking, 
	where $\alpha$ is contracted as 
	$\alpha \leftarrow \rho_{\mathrm{bt}} \alpha$ with 
	$\rho_{\mathrm{bt}} = 0.5$ until the sufficient 
	decrease condition
	\begin{equation}\label{eq:armijo}
		\small
		f_l^{(t)}\!\left(\Pi_{[P_{\min}^{\mathrm{p}},P_{\max}^{\mathrm{p}}]}\!\left(p_l^{(j)}-\alpha d^{(j)}\right)\right)
		\leq
		f_l^{(t)}(p_l^{(j)})-c_1\alpha\left|(f_l^{(t)})'(p_l^{(j)})\right|
	\end{equation}
	is satisfied with $c_1 = 10^{-4}$. To guarantee global 
	optimality over the bounded interval despite potential 
	nonconvexity, the converged Newton solution is compared 
	against both boundary values 
	$\{P_{\min}^{\mathrm{p}}, P_{\max}^{\mathrm{p}}\}$, 
	and the minimizer is selected.
	

	To accelerate convergence, UEs are sorted at each outer iteration by descending current QoS violation $g_k^{(t)}$, so that the most constrained UEs are updated first. Moreover, the regularization weight~$\omega$ is auto-tuned via 
	\begin{equation}
		\omega
		=
		\eta\,P_{\max}^{\mathrm{p}}
		\left\|
		\nabla_{\mathbf{p}^{\mathrm{p}}}
		\left(
		\frac{1}{K}\sum_{k=1}^{K}\mathrm{NMSE}_k(\mathbf{p}^{\mathrm{p}})
		\right)
		\Bigg|_{\mathbf{p}^{\mathrm{p}}=P_{\max}^{\mathrm{p}}\mathbf{1}_K}
		\right\|_1.
		\label{reg_weight}
	\end{equation}
	\vspace{-6mm}
	
	\begin{remark}
		The regularization 
		weight $\omega$ is an auto-calibrated scale for the linear 
		pilot power penalty in \eqref{eq:pilot_obj}. It is evaluated at the 
		conventional full-power pilot baseline 
		$\mathbf{p}^{\mathrm{p}}=P_{\max}^{\mathrm{p}}\mathbf{1}_K$, which 
		provides a reproducible reference point before any pilot power 
		update is performed. The gradient 
		$\nabla_{\mathbf{p}^{\mathrm{p}}}\left(K^{-1}\sum_k 
		\mathrm{NMSE}_k\right)$ has units of inverse power, and 
		multiplication by $P_{\max}^{\mathrm{p}}$ produces a dimensionless 
		scale consistent with the normalized power term 
		$\frac{1}{K P_{\max}^{\mathrm{p}}}\sum_k p_k^{\mathrm{p}}$ in the 
		objective. The dimensionless parameter $\eta$ then controls the 
		relative strength of the power penalty with respect to the NMSE 
		sensitivity at this reference point, so that a fixed value of 
		$\eta$ has a comparable interpretation across network realizations.
	\end{remark}
	
	The penalty parameter $\mu$ is escalated by $\mu_{\mathrm{scale}}>1$ when
	$\max_k g_k^{(t)}>\epsilon_{\mathrm{feas}}$. The complete procedure is summarized in Algorithm~\ref{alg:bcd_pilot}.
	
	\begin{algorithm}[t]
		\small
		\caption{BCD-based Pilot Power Allocation}
		\label{alg:bcd_pilot}
		\KwIn{$\{\mathbf{x}_m^\star\}$, 
			$\{\mathbf{R}_{k,m}\}$, $\boldsymbol{\pi}$, 
			$T_{\max}=3$, $J_{\max}=8$, 
			$\epsilon_{\mathrm{obj}}=10^{-6}$, 
			$\epsilon_{\mathrm{feas}}=10^{-6}$, 
			$\mu_0>0$, $\mu_{\mathrm{scale}}=2$, $\eta=0.1$}
		\KwOut{Optimized pilot powers 
			$(\mathbf{p}^{\mathrm{p}})^\star$}
		\BlankLine
		\textbf{Initialize:} 
		$\mathbf{p}^{\mathrm{p},(0)} =
		P_{\max}^{\mathrm{p}}\mathbf{1}_K$, 
		$\mu \leftarrow \mu_0$, $t \leftarrow 0$\;
		Auto-tune $\omega$ using the gradient-ratio rule above\;
		\Repeat{$|\widehat{\mathcal{L}}^{(t)}-\widehat{\mathcal{L}}^{(t-1)}|
			/|\widehat{\mathcal{L}}^{(t-1)}| < \epsilon_{\mathrm{obj}}$ 
			\textbf{or} $t > T_{\max}$}{
			{Update $\varepsilon_k^{(t)}=\varepsilon_k(\mathbf{p}^{\mathrm{p},(t)})$ via 
				Theorem~\ref{thm:nmse_threshold}, and recompute the reduced SINR 
				coefficients $\{A_k,B_k,C_k\}$ using $\mathbf{p}^{\mathrm{p},(t)}$ 
				and the fixed antenna positions, $\forall k$\;}
			Compute $g_k^{(t)}=\max(0,\mathrm{NMSE}_k-\varepsilon_k^{(t)})$ and
			$w_k^{(t)} = 1/K + 2\mu g_k^{(t)}$, $\forall k$\;
			Sort UEs by descending $g_k^{(t)}$\;
			\For{$l \in \text{sorted UE order}$}{
				\For{$j \leftarrow 1$ \KwTo $J_{\max}$}{
					Evaluate $\bigl(f_l^{(t)}\bigr)', \bigl(f_l^{(t)}\bigr)''$ via 
					\eqref{eq:gradient}--\eqref{eq:hessian}\;
					Set $d^{(j)}$ based on curvature 
					$\bigl(f_l^{(t)}\bigr)''$ versus $\delta_H$\;
					Determine $\alpha^{(j)}$ via Armijo 
					backtracking\;
					Update $p_l^{\mathrm{p}}$ via 
					\eqref{eq:newton_update}\;
				}
				$p_l^{\mathrm{p}} \leftarrow 
				\arg\min_{p \in \{p_l^{\mathrm{p}}, 
					P_{\min}^{\mathrm{p}}, 
					P_{\max}^{\mathrm{p}}\}} f_l^{(t)}(p)$\;
				Update $\boldsymbol{\Psi}_{k,m}$ and 
				$\mathrm{NMSE}_k$ for 
				$k \in \mathcal{P}_l$\;
			}
			Evaluate $\widehat{\mathcal{L}}^{(t)}(\mathbf{p}^{\mathrm{p},(t+1)};\mu)$\;
			\lIf{$\max_k g_k^{(t)} > \epsilon_{\mathrm{feas}}$}{
				$\mu \leftarrow \mu_{\mathrm{scale}} \cdot \mu$}
			$t \leftarrow t + 1$\;
		}
	\end{algorithm}
	\vspace{-2mm}
	\begin{proposition}\label{prop:bcd_conv}
		Fix an outer iteration $t$, together with the frozen thresholds
		$\{\varepsilon_k^{(t)}\}_{k=1}^{K}$ and weights
		$\{w_k^{(t)}\}_{k=1}^{K}$. Then the inner BCD sweep of
		Algorithm~\ref{alg:bcd_pilot} generates a sequence whose surrogate objective
		$\widehat{\mathcal{L}}^{(t)}(\mathbf{p}^{\mathrm{p}};\mu)$
		is monotonically non-increasing. Moreover, every limit point of the inner BCD sequence is a coordinate-wise minimum of
		$\widehat{\mathcal{L}}^{(t)}(\cdot;\mu)$ over
		$[P_{\min}^{\mathrm{p}}, P_{\max}^{\mathrm{p}}]^K$.
	\end{proposition}
	\vspace{-2mm}
	\begin{proof}
		For fixed $t$, the weights $\{w_k^{(t)}\}$ remain constant. Each per-UE update globally minimizes~\eqref{eq:bcd_sub} over $[P_{\min}^{\mathrm{p}}, P_{\max}^{\mathrm{p}}]$ via the boundary comparison, so $\widehat{\mathcal{L}}^{(t)}(\cdot;\mu)$ cannot increase. Coordinate-wise optimality of limit points follows from compactness and continuity~\cite[Prop.~2.7.1]{bertsekas1999}. 
	\end{proof}
	
	\subsubsection{Computational Complexity}
	\label{para:com_com_BCD}
	Each outer iteration involves $K$ per-UE updates. For 
	UE~$l$, the dominant costs are: (i)~recomputation of 
	LMMSE covariances for the co-pilot set at 
	$\mathcal{O}(|\mathcal{P}_l| \cdot 
	|\mathcal{M}_{\max}| \cdot N^3)$, and 
	(ii)~$J_{\max}$ Newton steps, each evaluating the 
	auxiliary trace quantities at 
	$\mathcal{O}(|\mathcal{P}_l| \cdot 
	|\mathcal{M}_{\max}| \cdot N^3)$. The total 
	per-outer-iteration complexity is 
	$\mathcal{O}(K(1 + J_{\max}) \cdot 
	|\mathcal{P}|_{\max} \cdot 
	|\mathcal{M}|_{\max} \cdot N^3)$, where 
	$|\mathcal{P}|_{\max} = K/\tau_p$ and 
	$|\mathcal{M}|_{\max} = \max_k |\mathcal{M}_k|$. 
	Since the per-UE subproblems are scalar with 
	analytically available gradient and Hessian, the 
	approach avoids generic solver overhead and scales 
	efficiently for dense CF mMIMO networks.
	
	\vspace{-3mm}
	
	\subsection{Practical Considerations and Implementation Aspects}
	\label{subsec:practical_considerations}
	
	We outline how the proposed optimization framework integrates with practical CF mMIMO deployments serving URLLC traffic. The discussion clarifies the operational timescales, the statistical-CSI acquisition pipeline, fronthaul signaling, the hardware realization of FA elements, and mobility-related considerations, all of which shape the applicability of the proposed framework.
	
	\subsubsection{Two-timescale operation}
	\label{subsubsec:two_timescale}
	
	The proposed framework follows a two-timescale architecture. FA 
	positions and pilot powers are optimized using long-term statistical 
	CSI, comprising the large-scale fading coefficients $\{\beta_{k,m}\}$ 
	and spatial correlation matrices $\{\mathbf{R}_{k,m}\}$. Hence, they are updated on the slow fading timescale over which these 
	statistics remain stable under the local scattering model in 
	\eqref{eq:correlation_model}~\cite{adhikary2013joint,bjornson2017massive}. 
	Within each coherence interval of length $\tau_c$, URLLC pilot and 
	data transmission proceed under the currently configured antenna 
	positions and pilot powers. Consequently, the delay constraint in 
	\eqref{eq:delay_constraint} governs the packet service process after 
	configuration, rather than requiring mechanical or port-level FA 
	reconfiguration within each per-packet latency budget $D_{\max,k}$. This separation is consistent 
	with the MMPP/M/1 model in Section~\ref{subsec:urllc_traffic}, where 
	the traffic process determines the required service rate, while the 
	physical-layer optimizer provides a statistical operating point that 
	supports this rate. In practice, FA reconfiguration is treated as a 
	slow-timescale configuration update, performed before a new 
	statistical operating period, so packet transmission proceeds under a 
	fixed antenna configuration.
	
	\subsubsection{Statistical CSI acquisition}
	\label{subsubsec:statistical_csi}
	
	The position and power updates in Algorithms~\ref{alg:sca_position} 
	and~\ref{alg:bcd_pilot} rely on statistical CSI, primarily 
	$\{\beta_{k,m}\}$ and $\{\mathbf{R}_{k,m}\}$, rather than instantaneous 
	small-scale channel realizations. These quantities vary over many 
	coherence intervals and can be estimated from long-term channel 
	observations, covariance tracking, or periodic 
	sounding~\cite{neumann2018covariance}. Once the statistics are 
	available, the joint optimization is executed on the slow timescale, 
	and the resulting configuration is reused for a large number of URLLC 
	packets. The optimization overhead is therefore amortized over many 
	coherence blocks and does not enter the per-packet latency budget.
	
	\subsubsection{Fronthaul overhead}
	\label{subsubsec:fronthaul}
	
	The fronthaul signaling associated with the proposed framework is 
	statistical rather than per-packet. APs do not forward instantaneous 
	channel realizations for antenna-position optimization. Instead, each 
	AP reports the covariance statistics of the UE links in its serving 
	set, namely the $N\times N$ matrices $\mathbf{R}_{k,m}$ for 
	$k\in\mathcal{K}_m$, together with large-scale fading information. 
	Since each covariance matrix has dimension $N\times N$ and the served 
	UE set is bounded by UC clustering, the statistical-CSI 
	exchange scales with the local cluster size rather than with the 
	total number of AP-UE pairs in the network. The optimization 
	therefore does not require additional per-packet fronthaul exchange 
	beyond the baseline centralized CF processing described in 
	Section~\ref{subsec:data_transmission}, provided that the statistical 
	updates are triggered only when the long-term propagation environment 
	changes significantly.
	
	\subsubsection{Discrete and switched-port implementations}
	\label{subsubsec:discrete_impl}
	
	The continuous-position model in \eqref{eq:position_constraints} 
	represents an ideal high-resolution FA implementation and provides a 
	tractable benchmark for position optimization. For switched-port or 
	pixelated FA architectures, the same framework applies with 
	$\mathcal{X}_m$ restricted to a predefined port grid. A practical 
	deployment can first compute the continuous SCA solution and then map 
	each element to the nearest feasible port while preserving the 
	ordering and minimum-spacing constraints in 
	\eqref{eq:position_constraints}. The resulting implementation loss is 
	controlled by the port spacing, since the correlation entries in 
	\eqref{eq:correlation_model} vary smoothly with the inter-element 
	distances. When candidate positions are restricted to a discretized 
	set, monotonicity can be enforced by accepting only grid-feasible 
	candidates that reduce the true penalized objective in 
	\eqref{eq:augmented_lagrangian}. The continuous-stationarity 
	statement in Proposition~\ref{prop:sca_convergence}, however, does 
	not directly apply to a finite grid without a separate discrete 
	optimality analysis. Quantifying the resulting loss for a specific 
	hardware topology requires the corresponding port layout and 
	switching constraints, and is therefore treated as an 
	implementation-specific refinement rather than a change to the 
	optimization principle.
	
	\subsubsection{Mutual coupling}
	\label{subsubsec:coupling}
	
	The minimum-spacing constraint $d_{\min}=\lambda/2$ in 
	\eqref{eq:position_constraints} is included to mitigate strong 
	electromagnetic coupling between neighboring elements and to maintain 
	a physically meaningful correlation model. The analysis in this paper 
	follows the standard covariance-based abstraction used in correlated 
	massive MIMO~\cite{adhikary2013joint,bjornson2017massive}, where 
	array geometry determines the correlation through inter-element 
	separations and angular spread. If a coupling-aware antenna model is 
	available, one representative extension replaces 
	$\mathbf{R}_{k,m}(\mathbf{x}_m)$ in \eqref{eq:correlation_model} with 
	a coupling-compensated covariance 
	$\widetilde{\mathbf{R}}_{k,m}(\mathbf{x}_m)=
	\mathbf{C}_m(\mathbf{x}_m)\mathbf{R}_{k,m}(\mathbf{x}_m)
	\mathbf{C}_m^{\mathsf{H}}(\mathbf{x}_m)$, where 
	$\mathbf{C}_m(\mathbf{x}_m)$ denotes the position-dependent coupling 
	matrix. The SCA and BCD structures remain applicable provided 
	$\widetilde{\mathbf{R}}_{k,m}$ is differentiable in $\mathbf{x}_m$ 
	and satisfies the same smoothness and boundedness conditions used in 
	the descent argument. The gradient expressions would then include 
	the additional derivatives of $\mathbf{C}_m(\mathbf{x}_m)$. 
	Thus, 
	mutual coupling primarily changes the covariance model supplied to 
	the optimizer, while the statistical-CSI-based SCA and BCD structure 
	remains applicable under the corresponding smoothness conditions.
	
	\subsubsection{Mobility and covariance stability}
	\label{subsubsec:mobility}
	
	The proposed design targets scenarios in which the spatial covariance 
	evolves slowly relative to the FA reconfiguration interval, such as 
	indoor industrial deployments, fixed wireless sensors, and 
	low-to-moderate mobility URLLC terminals~\cite{sarker2025NS}. In such 
	regimes, the nominal angles $\theta_{k,m}$ and angular spread 
	$\sigma_\theta$ in \eqref{eq:correlation_model} remain approximately 
	stationary across successive reconfiguration cycles. Under moderate 
	mobility, drift in $\{\mathbf{R}_{k,m}\}$ can be handled through 
	periodic refresh of the statistical CSI and warm-started re-execution 
	of Algorithms~\ref{alg:sca_position} and~\ref{alg:bcd_pilot} 
	initialized at the previous solution. Very fast mobility, where the 
	covariance changes on the same order as the statistical update or 
	reconfiguration period, would require predictive covariance tracking 
	or mobility-aware reconfiguration along the lines 
	of~\cite{DaiMovableCF_HST}, which is left for future work as a 
	natural extension of the proposed statistical-CSI framework.
	
	Taken together, these considerations position FA reconfiguration and 
	pilot power allocation as slow-timescale configuration steps that 
	reshape the covariance environment in which fast-timescale URLLC 
	transmission unfolds. The optimization results developed above 
	therefore apply to any CF mMIMO deployment consistent with this 
	two-timescale operating regime.
	
	\vspace{-4mm}
	\section{Numerical Analysis}
	\label{sec:simulations}
	\vspace{-2mm}
	\subsection{Simulation Setup}
	\label{subsec:simulation_setup}
	
	We consider a CF mMIMO network deployed over a $200\times200$~m$^2$ area with $M = 16$ APs and $K = 8$ single-antenna UEs, unless stated otherwise. Each AP is equipped with $N = 4$ FA elements positioned within an aperture of $W_m = 4\lambda$ with minimum inter-element spacing $d_{\min} = \lambda/2$. Both APs and UEs are distributed uniformly at random, with wrap-around applied to suppress boundary effects \cite{bjornson2017massive}. The system operates at carrier frequency $f_c = 3.5$~GHz with total bandwidth $B = 20$~MHz. Channels follow 3GPP
	Urban Microcell path loss~\cite{3gpp_38901} with 4~dB log-normal
	shadowing and angular spread $\sigma_\theta=10^\circ$. The maximum UE transmit power for both pilot and data is $20$~dBm, and the receiver noise figure is $9$~dB. The coherence interval spans $\tau_c = 200$ symbols with $\tau_p = 4$ orthogonal pilot sequences and pilots are  assigned randomly to all UEs. For AP-UE association, we adopt the approach from \cite{sarker_auction}. The URLLC traffic is modeled as a two-state MMPP as described in Section~\ref{subsec:urllc_traffic}, and in the simulations, we consider $\lambda_{\mathrm{ON},k}\in[50,200]$ pkt/s, $\xi_k \in [33, 100]$~s$^{-1}$, $\nu_k \in [8, 17]$~s$^{-1}$, $L_k \in [256, 512]$~bits, and $D_{\max,k}\in[1,5]$ ms, while the decoding- and delay-violation targets are fixed at $\epsilon_k^{\mathrm{dec}}=10^{-5}$ and $\epsilon_k^{\mathrm{delay}}=10^{-5}$. All results are averaged over 100 independent network realizations. 
	
	\vspace{-4mm}
	\subsection{Validation of Analytical Framework}
	\label{subsec:validation}
	
	Before evaluating the proposed optimization algorithm, we validate the analytical approximations underlying the problem formulation, focusing on physical-layer approximation accuracy. Specifically, we verify the trace product approximation in Lemma~\ref{lem:trace_approx}, the NMSE variation bound in Proposition~\ref{prop:nmse_variation}, and the quadratic SINR-NMSE relationship derived in Section~\ref{subsec:sinr_nmse}. Notably, the tested scenarios include non-orthogonal pilot reuse, and the approximation accuracy remains high in these contaminated operating conditions, indicating that the pointwise-evaluated coefficient $C_k$ captures the coherent pilot-sharing contribution sufficiently well for SINR prediction and NMSE-threshold computation in the operating regimes of interest. Unless otherwise stated, we use 100 setups yielding 800 UE samples; for the computationally intensive traffic-dependent validation, we use 50 setups (400 samples).
	\subsubsection{Trace Product Approximation Bound}
	\label{subsubsec:trace_bound}
	Fig.~\ref{fig:trace_approx} validates Lemma~\ref{lem:trace_approx} by plotting
	the relative trace product approximation error versus pilot
	SNR for different angular spreads~$\sigma_\theta$. The
	actual error (markers) closely follows the theoretical bound
	(dashed lines) across all configurations. As predicted, the error decreases with increasing pilot SNR and with larger angular spread, which reduces the condition number $\kappa$. For $\rho_{k,m} \geq 10$~dB and $\sigma_\theta \geq 10^\circ$, the  error falls below 1\%, confirming the validity conditions stated after Lemma~\ref{lem:trace_approx}.
	\begin{figure}[t]
		\centering
		\includegraphics[scale = 0.4]{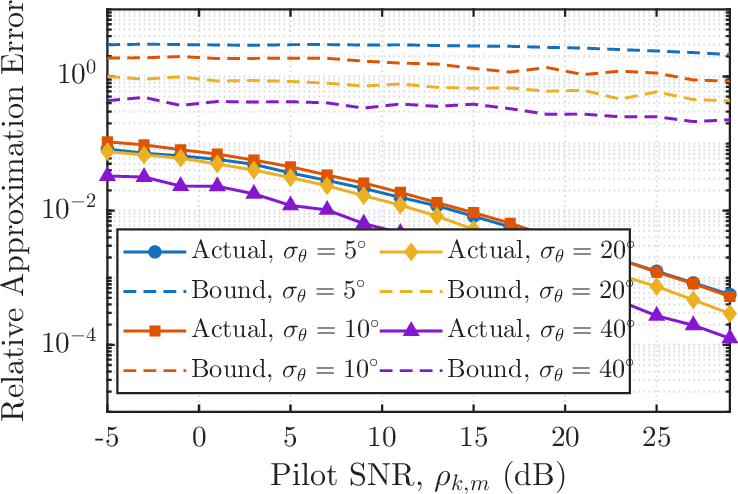}
		\caption{Trace product approximation error versus pilot SNR for different angular spreads. Markers show actual error; dashed lines show theoretical bounds from Lemma~\ref{lem:trace_approx}.}
		\label{fig:trace_approx}
	\end{figure}
	
	\subsubsection{NMSE Variation Bound and Sensitivity}
	\label{subsubsec:nmse_variation}
	Fig.~\ref{fig:nmse_variation}(a) compares the empirical per-AP NMSE variation $\Delta_k\triangleq \max_{m,m'\in\mathcal{M}_k} |\mathrm{NMSE}_{k,m}-\mathrm{NMSE}_{k,m'}|$ with the bound $\delta_k$ under prescribed intra-cluster large-scale-fading spreads. The bound holds in all $2400$ controlled realizations. Both $\Delta_k$ and $\delta_k$ generally decrease with pilot SNR, although their rates depend strongly on the weakest serving link. For spreads of $3$ and $10$~dB, the mean bound falls below $0.05$ at average-link pilot SNRs of $15$ and $20$~dB, respectively. At spreads of $20$ and $30$~dB, the weakest-link pilot SNR remains low over the considered average-link SNR range, resulting in larger per-AP NMSE variation. The bound therefore provides a valid characterization across the tested heterogeneity range, while the uniform-NMSE approximation is most accurate when all serving links retain sufficient estimation quality.
	
	Fig.~\ref{fig:nmse_variation}(b) quantifies the corresponding SINR sensitivity by comparing a heterogeneous surrogate that retains $\alpha_{k,m}=1-\mathrm{NMSE}_{k,m}$ with the uniform surrogate using the weighted effective factor $\alpha_k=1-\mathrm{NMSE}_k$. The effective-NMSE definition preserves the desired-signal and effective-noise terms, while the residual difference arises primarily from the beamforming-uncertainty and noncoherent-interference terms. Among $15095$ practical-regime samples, the median and $90$th-percentile absolute SINR deviations are $0.017$ and $0.286$~dB, respectively. For $\Delta_k<0.3$, which includes approximately $69\%$ of these samples, the $90$th-percentile deviation remains below $0.283$~dB. Larger deviations are concentrated in the upper tail where the serving-AP NMSE values are strongly nonuniform. Hence, the approximation is reliable across the dominant practical regime, while configurations with exceptionally large per-AP NMSE disparity require explicit heterogeneous treatment.

	\begin{figure}[t]
		\centering
		\subfloat[]{\includegraphics[width=0.48\columnwidth]{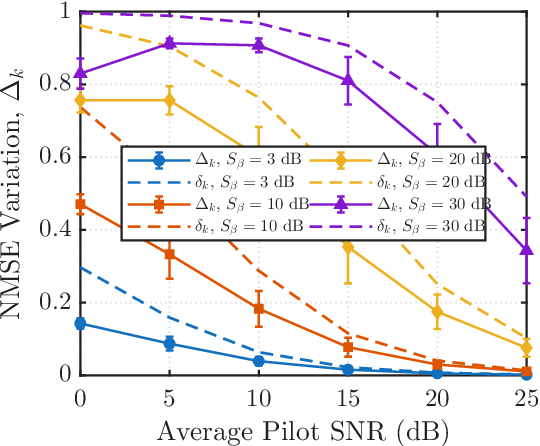}}
		\hfill
		\subfloat[]{\includegraphics[width=0.48\columnwidth]{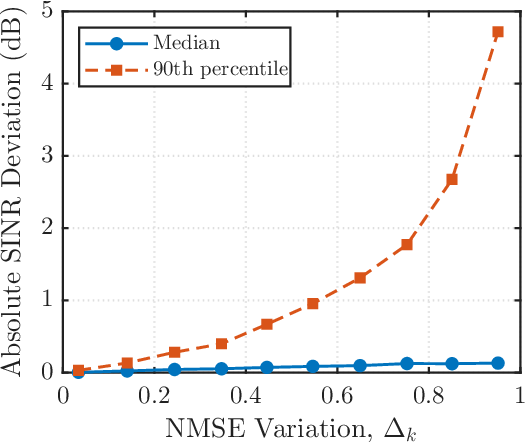}}
		\caption{{Validation of Assumption~\ref{ass:uniform_nmse}: (a) empirical NMSE variation $\Delta_k$ and bound $\delta_k$ versus average-link pilot SNR; (b) median and $90$th-percentile absolute SINR deviation versus $\Delta_k$ in the practical regime ($\mathrm{SINR}>-10$~dB and $\mathrm{NMSE}<0.5$).}}
		\label{fig:nmse_variation}
	\end{figure}
	
	\subsubsection{SINR-NMSE Relationship}
	To validate the quadratic SINR-NMSE relationship, Fig.~\ref{fig:sinr_validation}(a) compares exact SINR from \eqref{eq:sinr} and
	approximate SINR from~\eqref{eq:sinr_quadratic} across 800
	randomly generated UE configurations, with color indicating
	effective NMSE. Samples outside the practical regime
	($\mathrm{SINR}\leq -10$\,dB or $\mathrm{NMSE}\geq 0.5$)
	constitute only 3.5\% of samples (gray points).
	Fig.~\ref{fig:sinr_validation}(b) shows the error distribution,
	confirming substantially lower errors within the practical
	regime. Table~\ref{tab:approx_accuracy} summarizes accuracy
	across operating regimes, showing the median absolute error is
	0.02--0.03\,dB in all cases, with over 96\% of samples
	within 10\% error when $\mathrm{NMSE}<0.1$. Consistent with Fig.~\ref{fig:nmse_variation}(b), the uniform-NMSE substitution produces median and $90$th-percentile absolute SINR deviations of $0.017$ and $0.286$~dB, respectively, confirming the limited impact of Assumption~\ref{ass:uniform_nmse} over the dominant practical regime. These results further confirm that the reduced coefficient $C_k$ provides an accurate effective representation of the coherent pilot-sharing contribution for the purpose of SINR prediction and NMSE-threshold computation in the operating regimes of interest.
	\begin{figure}[t]
		\centering
		\subfloat[Exact vs. approximate SINR]{\includegraphics[width=0.48\columnwidth]{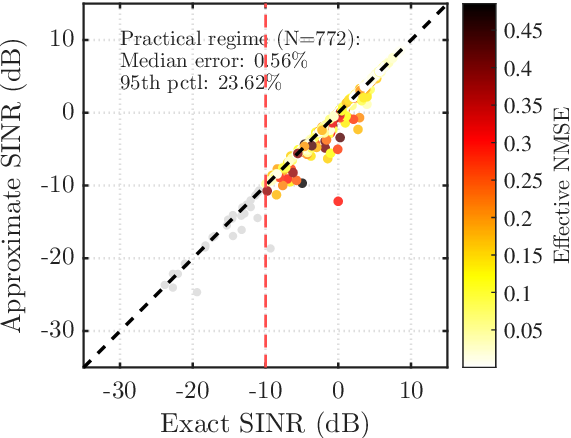}}
		\hfill
		\subfloat[Error distribution]{\includegraphics[width=0.48\columnwidth]{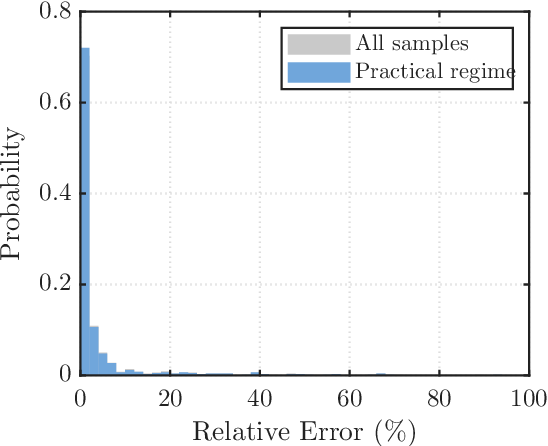}}
		\caption{Validation of SINR-NMSE approximation: (a) scatter plot comparing exact and approximate SINR values, with color indicating NMSE (gray points lie outside practical regime); (b) distribution of relative approximation error for all samples (gray) and practical operating regime (blue).}
		\label{fig:sinr_validation}
	\end{figure}
	
	\begin{table}[t]
		\centering
		\caption{SINR Approximation Accuracy by Operating Regime}
		\label{tab:approx_accuracy}
		\begin{tabular}{lcccc}
			\hline
			\textbf{Operating} & \textbf{Samples} & \textbf{Median} & \textbf{Median Abs.} & \textbf{Accur. within} \\
			\textbf{Regime} & (\%) & Error (\%) & Error (dB) &  $<10$\% (\%) \\
			\hline
			All samples & 100 & 0.70 & 0.03  & 89.6 \\
			Practical$^*$ & 96.5 & 0.67 & 0.03  & 90.2 \\
			NMSE $< 0.1$ & 89.4 & 0.57 & 0.02  & 96.8 \\
			SINR $> 0$~dB & 73.1 & 0.51 & 0.02  & 93.7 \\
			\hline
			\multicolumn{5}{l}{\footnotesize $^*$Practical regime: SINR $> -10$~dB and NMSE $< 0.5$.}
		\end{tabular}
	\end{table}
	\subsubsection{Validation at URLLC Operating Points}
	\label{subsubsec:urllc_op_validation}
	We next validate the approximation at traffic-derived operating
	points using 50 setups (400 UE samples). For each setup, the MMPP/M/1 queue is used to obtain the minimum required rate $R_k^{\min}$ under the delay constraint. Here, feasibility is defined by $\mathrm{SINR}_k \ge \bar{\gamma}_k$, where $\bar{\gamma}_k$ is obtained from the MMPP-based minimum rate requirement $R_k^{\min}$ via FBL inversion. The SINR
	targets from MMPP/M/1 queueing analysis span
	$[-6.62,\,-1.56]$\,dB with median $-3.96$\,dB, and 86.5\%
	of samples are feasible ($\mathrm{SINR}_k\geq\bar{\gamma}_k$).
	Fig.~\ref{fig:urllc_op_validation}(a) shows that feasible
	samples align closely with the diagonal, while
	Fig.~\ref{fig:urllc_op_validation}(b) reveals a median
	relative error of 0.51\% (0.022\,dB absolute) in the feasible
	region versus 3.52\% in the infeasible region. The
	approximation is thus most accurate precisely where QoS
	constraints are active, ensuring that derived NMSE thresholds
	faithfully reflect the underlying SINR requirements.
	
	\begin{figure}[t]
		\centering
		\subfloat[Exact vs. approximate SINR]{\includegraphics[width=0.48\columnwidth]{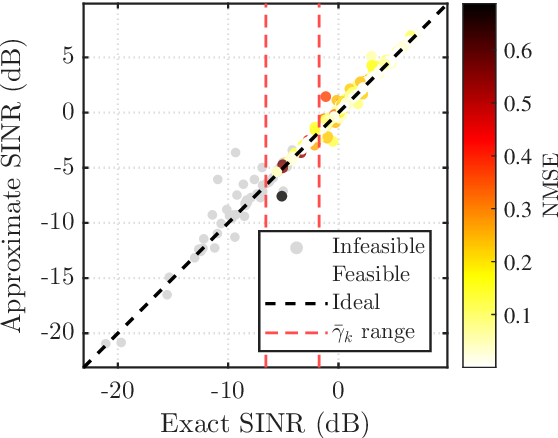}}
		\hfill
		\subfloat[Error distribution]{\includegraphics[width=0.48\columnwidth]{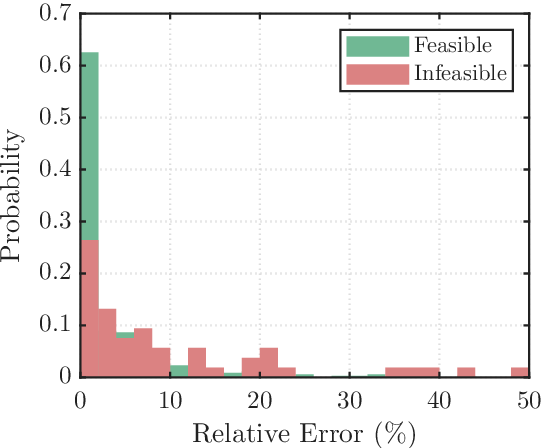}}
		\caption{Traffic-dependent validation at URLLC operating points: (a) scatter plot comparing exact and approximate SINR values, with color indicating NMSE (gray points are infeasible); (b) distribution of relative approximation error for feasible (green) and infeasible (red) samples.}
		\label{fig:urllc_op_validation}
	\end{figure}

	\vspace{-3mm}
	\subsection{Performance of Proposed Antenna Positioning Scheme}
	\label{subsection:performance_positioning}
	
	In this section, we evaluate the performance of the proposed SCA-based positioning scheme (Algorithm~\ref{alg:sca_position}) against five baselines across different parameters. The fixed uniform array (FUA) serves as the reference, employing conventional $\lambda/2$ spacing. Random fluid antenna (RFA) averages over 50 random feasible configurations, while channel-aware spacing (CAS) and maximum spacing (MS) represent heuristic approaches that adapt element positions based on angular statistics and aperture maximization, respectively. Lastly, particle swarm optimization (PSO) with a swarm of $P_s = 30$ particles evolved over $G = 100$ generations provides a global search benchmark. All considered schemes adopt full pilot power transmission and random pilot assignment. In what follows, we investigate the performance of Algorithm~\ref{alg:sca_position} under varying array size, aperture size, and network load across NMSE reduction and feasibility metrics.
	
	{Before presenting the positioning results, we clarify how channel-estimation error should be interpreted in the following performance metrics. The proposed schemes act on the LMMSE channel-estimation error 
		through the NMSE in \eqref{eq:nmse_effective}. A lower per-UE NMSE 
		improves the effective SINR through the quadratic-fractional mapping 
		in \eqref{eq:sinr_quadratic}, increases the FBL service rate in 
		\eqref{eq:fbl_rate}, and thereby improves the delay-reliability 
		margin of the MMPP/M/1 queue in 
		Section~\ref{subsec:urllc_traffic}. The QoS feasibility metric 
		reported below therefore indicates whether the estimation quality is 
		sufficient to meet the SINR target implied by the FBL and 
		queueing-level delay analysis. The optimization inputs 
		$\{\beta_{k,m}\}$ and $\{\mathbf{R}_{k,m}\}$ are long-term statistics 
		estimated from channel observations. Finite-sample uncertainty in 
		these statistics would perturb the NMSE values and the corresponding 
		thresholds, and a dedicated robustness study under imperfect 
		statistical CSI is left for future work.}
	\subsubsection{Impact of Array Size}
	\label{subsubsection:array_size}
	
	Fig.~\ref{fig:NMSEreduction_vs_N} presents NMSE reduction compared to FUA and feasibility performance for different array sizes. The left panel reveals that the proposed SCA scheme consistently achieves significant NMSE reduction, improving by up to $22\%$ as $N$ increases. This performance closely matches PSO, which attains up to $21.4\%$ reduction, confirming that SCA effectively exploits the enlarged optimization space afforded by additional antenna elements. The near-optimal performance of SCA stems from the gradient-based
	descent exploiting the differentiable NMSE structure, adaptive
	proximal regularization ensuring stable convergence, and progressive
	penalty escalation balancing NMSE minimization against QoS
	satisfaction (Section~\ref{subsec:position_opt}). In contrast, the heuristic schemes exhibit substantial degradation relative to FUA, with MS and CAS increasing NMSE by up to $2.6\times$ and $2.5\times$, respectively, while RFA degrades performance by up to $2.2\times$. This degradation arises because spreading elements disrupts the
	strong spatial correlation that FUA's minimum spacing preserves,
	undermining LMMSE coherent combining gain, while CAS optimizes for a
	criterion misaligned with estimation quality,  established in Section~\ref{subsec:channel_model}.
	
	
	
	The right panel of Fig.~\ref{fig:NMSEreduction_vs_N} shows QoS
	feasibility, defined as the fraction of UEs with $\mathrm{NMSE}_k\leq\varepsilon_k$ derived from the target SINR $\bar{\gamma}_k$ via Theorem~\ref{thm:nmse_threshold}. All schemes achieve comparable feasibility, increasing from approximately $73\%$ to $92.5\%$ as $N$ grows, aligning with the analytical relationship established in Section~\ref{subsubsec:urllc_op_validation}. This confirms that feasibility is governed by fundamental channel conditions rather than positioning strategy. For feasible UEs, the proposed SCA scheme translates its NMSE reduction into higher achieved SINR and rate, validating its practical utility for FA-enabled CF mMIMO systems.
	\begin{figure}[tb]
		\centering
		\includegraphics[scale=0.5]{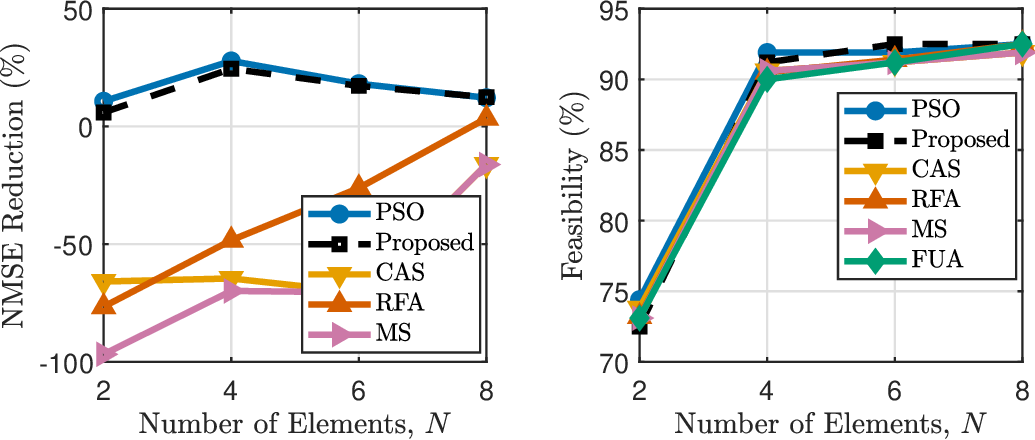}
		\caption{NMSE reduction relative to FUA (\%, left) and feasibility probability (\%, right) versus \(N\) for \(M=16\), \(K=8\), \(\tau_p=4\), and \(W_m=4\lambda\).}
		\label{fig:NMSEreduction_vs_N}
	\end{figure}
	\subsubsection{Impact of Aperture Size}
	\label{subsubsection:aperture_size}
	
	Fig.~\ref{fig:NMSEreduction_vs_Wm} reports the effect of the aperture size $W_m$ on the NMSE reduction relative to FUA and feasibility performance. The left panel reveals that both SCA and PSO peak
	at $W_m=4\lambda$ (12.7\% and 14.5\% reduction, respectively), with gains declining to 2.4\% and 3.6\% at $W_m=8\lambda$. This non-monotonic behavior occurs because larger apertures expand the feasible region without adding
	elements, so that favorable correlation structures become harder to identify and excessive spacing exceeds the decorrelation distance. Heuristic schemes degrade by up to $2.9\times$ at $W_m=8\lambda$, consistent with the mechanisms identified in Section~\ref{subsubsection:array_size}. The right panel shows that feasibility remains stable at 85\%--92\% across all aperture sizes, confirming that moderate apertures around $4\lambda$ provide the best tradeoff between optimization flexibility and solution quality.
	\begin{figure}[tb]
		\centering
		\includegraphics[scale=0.5]{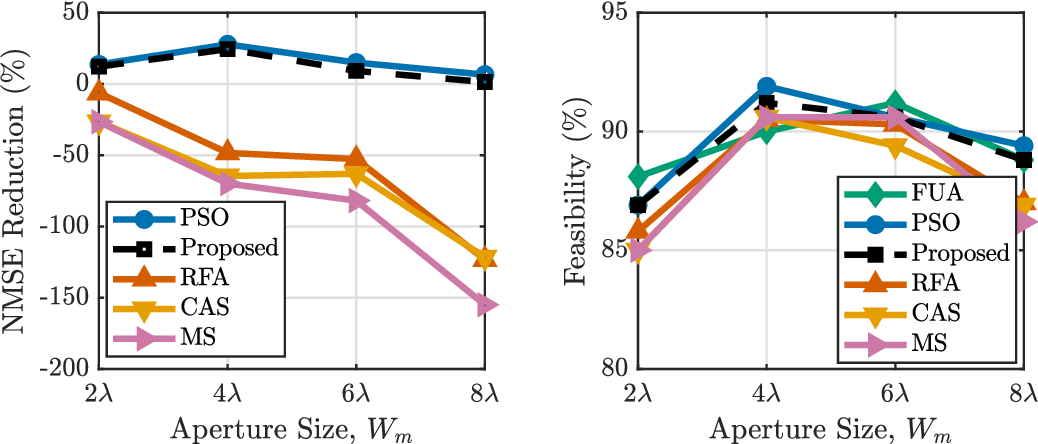}
		\caption{NMSE reduction relative to FUA (\%, left) and feasibility probability (\%, right) versus \(W_m\) for \(M=16\), \(K=8\), \(\tau_p=4\), and \(N=4\).}
		\label{fig:NMSEreduction_vs_Wm}
	\end{figure}
	\subsubsection{Impact of Network Load}
	\label{subsubsection:network_load}
	
	Fig.~\ref{fig:NMSEreduction_vs_K} demonstrates NMSE reduction relative to FUA and feasibility performance for varying UE density. The left panel of Fig.~\ref{fig:NMSEreduction_vs_K} shows that PSO achieves the highest reduction gain up to $17.8\%$, while the proposed SCA scheme attains a comparable gain of $13.2\%$. The heuristic schemes continue to exhibit degradation relative to FUA, with MS increasing NMSE by up to $2.3\times$ at low load. Importantly, gains diminish as $K$ increases for all schemes due to the increasing  co-pilot interference which dominates the estimation error under fixed $\tau_p$. Under such conditions, Algorithm~\ref{alg:sca_position} provides limited benefit since the contamination term in $\boldsymbol{\Psi}_{k,m}$ from~\eqref{eq:psi_matrix} masks the correlation structure improvements that SCA exploits. The right panel of Fig.~\ref{fig:NMSEreduction_vs_K} indicates a sharp feasibility decline from approximately 96\% at $K=8$ to 31--38\% at $K=32$, driven by pilot contamination elevating the NMSE floor and bandwidth limitations tightening per-UE SINR thresholds. These results indicate that FA positioning yields the greatest benefit under moderate loads, while heavily loaded scenarios require complementary pilot or bandwidth provisioning.	
	\begin{figure}[tb]
		\centering
		\includegraphics[scale=0.5]{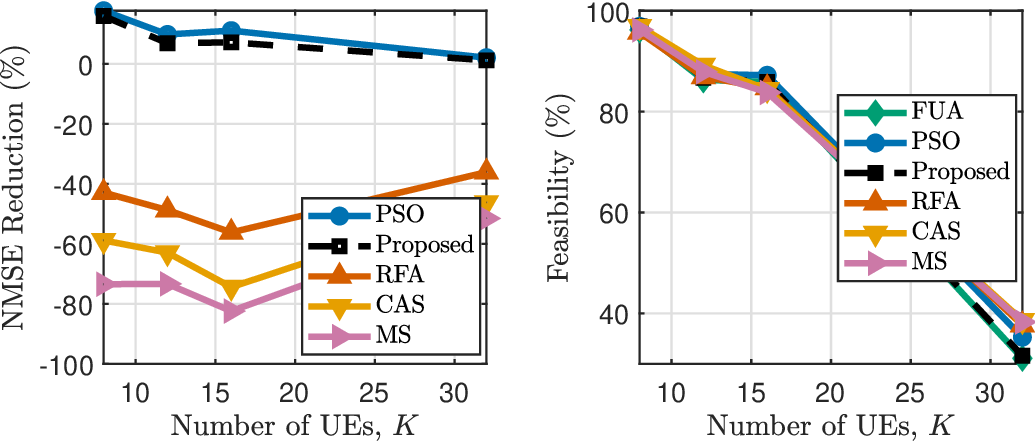}
		\caption{NMSE reduction relative to FUA (\%, left) and feasibility probability (\%, right) versus \(K\) for \(M=32\), \(\tau_p=4\), \(N=4\), and \(W_m=4\lambda\).}
		\label{fig:NMSEreduction_vs_K}
	\end{figure}

	Fig.~\ref{fig:runtime_vs_K} compares the runtime of the proposed SCA scheme against PSO under varying network load. The proposed scheme consistently achieves a speedup of $7$--$8\times$ across all evaluated loads compared to PSO. This computational advantage stems from the per-iteration complexity established in Section~\ref{para:comput_SCA}, where the gradient-based updates in Algorithm~\ref{alg:sca_position} require $\mathcal{O}(I_{\max} KMN^3)$ operations compared to $\mathcal{O}(P_sG \cdot KMN^3)$ for PSO with population size $P$ and $G$ generations. Notably, both schemes exhibit runtime growth with increasing $K$, as the gradient computation in SCA requires evaluating LMMSE covariances and their derivatives across all UEs and their respective pilot-sharing sets. Combined with the near-optimal NMSE reduction demonstrated in Section~\ref{subsubsection:network_load}, the proposed scheme offers a favorable tradeoff between solution quality and computational efficiency, making it suitable for practical deployment in CF mMIMO systems where periodic position reconfiguration is required.
	\begin{figure}[tb]
		\centering
		\includegraphics[scale= 0.5]{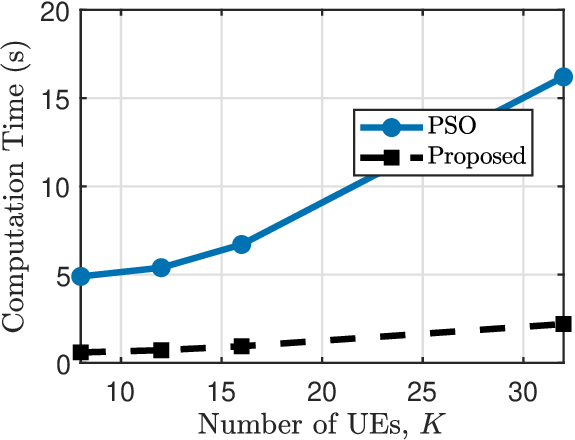}
		\caption{Runtime comparison between proposed and PSO schemes for varying $K$ under \(M=32\), \(\tau_p=4\), \(N=4\), and \(W_m=4\lambda\).}
		\label{fig:runtime_vs_K}
	\end{figure}
	\vspace{-3mm}
	\subsection{Performance of Proposed Pilot Power Allocation Scheme}
	\label{subsec:performance_power}

	\subsubsection{Outer AO Behavior}
	\label{subsubsec:convergence}
	Before evaluating the pilot power allocation schemes in detail, we assess the effectiveness of repeated AO between the proposed SCA-based antenna positioning algorithm (Algorithm~\ref{alg:sca_position}) and the BCD-based pilot power allocation algorithm (Algorithm~\ref{alg:bcd_pilot}) in providing further NMSE reduction beyond the sequential SCA--BCD design. In Table~\ref{tab:ao_cycle_comparison}, Cycle~1 denotes one SCA position update followed by one BCD pilot power update, while the subsequent cycles repeat these two blocks until the outer stopping condition or the maximum cycle count is reached. For setups that terminate early, the last available operating point is carried forward when computing the cycle-level averages. The first additional alternation reduces the aggregate mean NMSE by only $0.63\%$ relative to Cycle~1. Across Cycles~2--5, the aggregate mean NMSE remains within the narrow range $[0.024915,0.024952]$, and the maximum reduction relative to the sequential implementation remains below $0.64\%$. In contrast, the median cumulative optimization runtime becomes approximately three times that of the sequential implementation, while QoS feasibility decreases by $1.25$ percentage points. The marginal NMSE reduction is therefore not commensurate with the additional computational cost and does not improve QoS satisfaction. Accordingly, we retain the sequential SCA--BCD implementation in the remaining comparisons as the preferred performance-complexity operating point.
	\begin{table}[t]	
		\centering
		\caption{{Outer-AO Behavior with $M = 16$ APs, $K = 8$ UEs, $W_m = 4\lambda$, and $\tau_p = 4$.}}
		\label{tab:ao_cycle_comparison}
		\footnotesize
		\setlength{\tabcolsep}{3.5pt}
		\renewcommand{\arraystretch}{1.08}
		\begin{tabular}{ccccc}
			\hline
			\textbf{Cycle} & \textbf{Avg. NMSE} & \textbf{Change} & 
			\textbf{Feas.} & \textbf{Cum. time} \\
			& & \textbf{(\%)} & \textbf{(\%)} & \textbf{(s)} \\
			\hline
			1 & 0.025075 & 0.000 & 87.50 & 0.66 \\
			2 & 0.024917 & 0.631 & 86.25 & 1.33 \\
			3 & 0.024952 & 0.491 & 86.25 & 2.00 \\
			4 & 0.024915 & 0.637 & 86.25 & 2.00 \\
			5 & 0.024952 & 0.491 & 86.25 & 2.00 \\
			\hline
		\end{tabular}
	\end{table}
	
	\subsubsection{Large-Scale Tail and QoS Performance}
	\label{subsubsec:tail_estimation}
	Since URLLC reliability can be determined by unfavorable UEs rather than average performance alone, we next examine whether the average-NMSE gains of the selected sequential design extend to the upper tail and unfavorable-UE estimation quality. We consider a large-scale deployment with $M=200$, $K=100$, $N=4$, $\tau_p=20$, $L=16$, $B=120$~MHz, and a $500\times500$~m$^2$ area. Compared with the default configuration, this setting preserves $K/M=0.5$, while increasing the pilot-reuse ratio from $K/\tau_p=2$ to $5$ and reducing the per-UE bandwidth from $2.5$ to $1.2$~MHz. It therefore provides a more demanding large-scale evaluation rather than a direct replication of the default operating conditions. 
	
	Table~\ref{tab:large_scale_tail} reports the average, pooled $95$th-percentile, and pooled maximum NMSE, together with the empirical QoS-violation probability, aggregate pilot power, and median optimization runtime. The comparison includes FUA with full pilot power (FUA + Full), SCA positioning with full pilot power (SCA + Full), BCD pilot power allocation under the FUA geometry (FUA + BCD), and the sequential SCA--BCD design (SCA + BCD). These configurations represent the fixed, position-only, power-only, and joint designs, respectively. Relative to FUA with full pilot power, SCA--BCD reduces the average, $95$th-percentile, and maximum NMSE by $24.24\%$, $31.9\%$, and $18.5\%$, respectively. Moreover, SCA--BCD outperforms both the position-only and power-only configurations, confirming that antenna positioning and pilot power allocation provide complementary rather than interchangeable gains. Hence, the estimation-quality improvement extends beyond the average UE to the upper tail. The joint design also maintains the QoS-feasibility level under the more bandwidth-limited and pilot-reuse-intensive large-scale configuration, while reducing aggregate pilot power. Its average runtime is approximately one minute and is dominated by the positioning stage. Since the optimization is performed during periodic statistical-CSI updates rather than within the URLLC packet-service interval, this result supports the practical use of the sequential framework in large-scale deployments.
	
	Having established the scalability and tail behavior of the single-pass framework at $K=100$, we use the configuration mentioned in Section~\ref{subsec:simulation_setup} in the remaining experiments to enable broader Monte Carlo evaluation, parameter sweeps, and comparisons with pilot power baselines.
	
	\begin{table}[t]	
		\centering
		\caption{{Large-Scale Tail, QoS, and Computational Performance.}}
		\label{tab:large_scale_tail}
		\setlength{\tabcolsep}{2.0pt}
		\renewcommand{\arraystretch}{1.08}
		\begin{tabular}{lcccccc}
			\hline
			\textbf{Scheme}
			& \textbf{Avg.}
			& \textbf{95th}
			& \textbf{Max.}
			& \textbf{Viol.}
			& \textbf{Power}
			& \textbf{Time} \\
			& \textbf{NMSE}
			& \textbf{pct.}
			& \textbf{NMSE}
			& \textbf{prob.}
			& \textbf{(dBm)}
			& \textbf{(s)} \\
			\hline
			FUA + Full & 0.0257 & 0.1162 & 0.5031
			& 0.194 & 40.00 & -- \\
			SCA + Full & 0.0236 & 0.1055 & 0.4917
			& 0.182 & 40.00 & 52.86 \\
			FUA + BCD  & 0.0210 & 0.0836 & 0.4093
			& 0.192 & 39.17 & 8.99 \\
			SCA + BCD  & 0.0195 & 0.0791 & 0.4098
			& 0.182 & 38.28 & 64.36 \\
			\hline
		\end{tabular}
	\end{table}

	\subsubsection{Estimation Quality--Power Tradeoff}
	\label{subsubsec:estimation}
	To evaluate the proposed BCD-based pilot power allocation (Algorithm~\ref{alg:bcd_pilot}), we compare it against four benchmarks using the fixed antenna positions obtained from Algorithm~\ref{alg:sca_position}. Full power transmission (Full) sets 
	$p_k^{\mathrm{p}} = P_{\max}^p$ for all UEs and serves as the 
	baseline. Fractional programming (FP) reformulates the 
	sum-NMSE objective in~\eqref{eq:objective} as a 
	sum-of-ratios problem solved via the quadratic 
	transform~\cite{Sarker_RL_2026}. The zero-NMSE equilibrium method (ZEM) derives per-UE power from the $\mathrm{NMSE}_k \to 0$ condition 
	in~\eqref{eq:nmse_effective}, setting 
	$\Upsilon_k = \Gamma_k$ in~\eqref{eq:Upsilon_k} to 
	obtain a closed-form power expression that is 
	iteratively refined with damped fixed-point 
	updates~\cite{Sark2301:Pilot}. The exhaustive search 
	(Exhaustive) decomposes the $K$-dimensional problem into 
	$\tau_p$ independent subproblems by exploiting pilot-group 
	separability, solving each via grid search over $G = 200$ 
	power levels to obtain the global optimum up to 
	discretization error. {Note that the NMSE-to-QoS interpretation introduced in 
		Section~\ref{subsection:performance_positioning} also applies to 
		the pilot power allocation results reported below.}
	
	Fig.~\ref{fig:pareto} presents the Pareto-optimal tradeoff 
	between channel estimation quality and pilot energy consumption 
	enabled by the proposed scheme. Fig.~\ref{fig:pareto}(left) depicts the mean NMSE versus total pilot power for the proposed scheme, parameterized by the 
	regularization weight $\eta \in [0, 5]$, alongside the 
	benchmark methods appearing as single operating points. The 
	shaded region indicates $\pm 1$ standard deviation across 
	network realizations. At $\eta = 0$ (pure NMSE minimization), 
	Proposed~($\omega = 0$) achieves 30.6\% NMSE reduction 
	relative to Full, matching Exhaustive within 0.3 percentage 
	points. A penalty-free operating region emerges for the proposed scheme with 
	$\eta \leq 0.05$, where up to 33.5\% pilot power reduction 
	is attained with less than 0.1\% NMSE degradation relative 
	to the $\omega = 0$ solution. This near-lossless saving arises because $g_l^{\mathrm{self}}$ and
	$g_l^{\mathrm{cross}}$ in~\eqref{eq:gradient} nearly cancel at the
	NMSE optimum, leaving the power regularization term free to reduce
	energy without perturbing estimation quality. The auto-tuned operating 
	point at $\eta = 0.1$ achieves 38.8\% power reduction with 
	only 0.3\% NMSE penalty, positioned at the natural knee of 
	the Pareto frontier where marginal power savings begin to 
	incur measurable estimation loss. At $\eta = 0.2$, 
	Proposed~($\omega > 0$) strictly dominates ZEM by 
	simultaneously attaining 23.2\% and approximately 4\% lower NMSE  
	and pilot power, respectively, confirming 
	that the coordinate-wise gradient structure in 
	Section~\ref{sec:bcd_pilot} extracts superior estimation--power 
	tradeoffs compared to the ZEM scheme.
	
	Fig.~\ref{fig:pareto}(right) quantifies the NMSE gain and power 
	reduction as functions of $\eta$, with the auto-tuned default 
	$\eta = 0.1$ marked by the vertical dashed line. The NMSE 
	gain relative to Full remains above 26\% for 
	$\eta \leq 1.0$ and decreases to 13.1\% only at 
	$\eta = 5.0$, confirming robustness of the proposed 
	formulation across a wide regularization range. Power 
	reduction increases monotonically from 27.4\% to 86.4\%, 
	providing continuous and predictable control over the 
	estimation--energy operating point through the single parameter $\eta$ that controls the ratio of power cost to NMSE gradient magnitude, 
	as described in Section~\ref{sec:bcd_pilot}. This 
	tunability is particularly relevant for energy-constrained 
	URLLC deployments, where pilot transmission energy can be reduced while satisfying the probabilistic delay constraints in~\eqref{eq:delay_constraint} without compromising the channel estimation quality. 
	
	\begin{figure}[tb]
		\centering
		\includegraphics[scale=0.4]{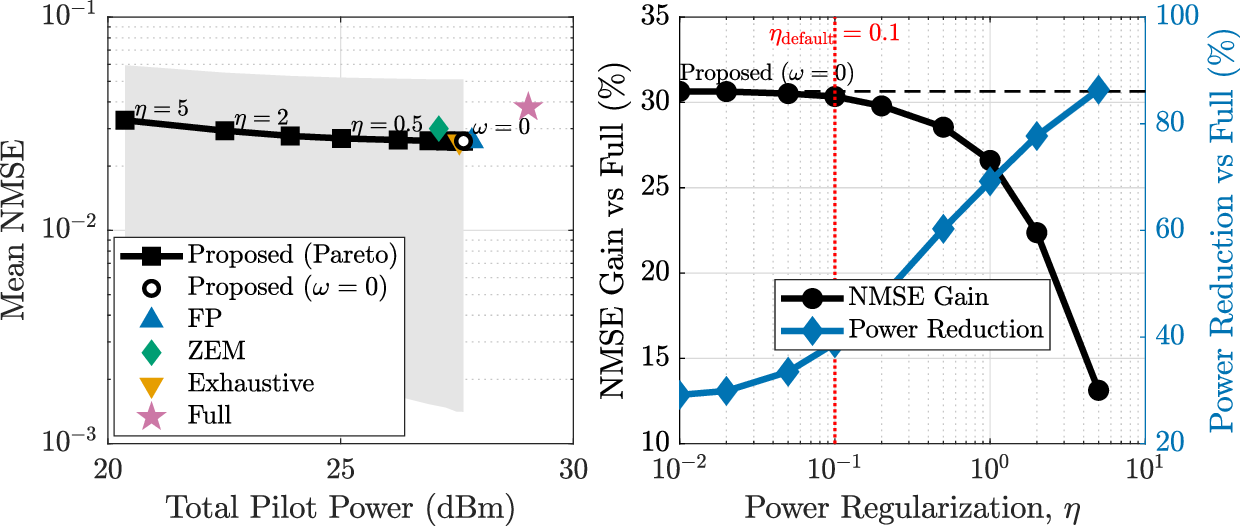}
		\caption{NMSE--power Pareto tradeoff for pilot power allocation
			with $M=16$, $K=8$, $N=4$, $\tau_p=4$, and $W_m=4\lambda$.}
		\label{fig:pareto}
	\end{figure}
	\subsubsection{Impact of UE Density}
	
	Fig.~\ref{fig:nmse_power_vs_K} evaluates pilot power allocation 
	performance for varying UE density. The left panel shows that both proposed variants achieve NMSE 
	reduction within 0.5 percentage points of the Exhaustive 
	optimum across all loads. At $K = 8$ ($K/\tau_p = 2$), 
	Proposed~($\omega = 0$) attains 30.7\% gain, matching FP 
	(30.7\%) and Exhaustive (31.0\%). As $K$ increases to 12, 
	all optimization-based methods experience a sharp decline to 
	approximately 17\%, after which the gains stabilize through 
	$K = 20$. This transition reflects the growing dominance of 
	pilot contamination in $\boldsymbol{\Psi}_{k,m}$ 
	from~\eqref{eq:psi_matrix}, where the co-pilot 
	interference contributions scale with $K/\tau_p$ 
	and progressively mask the estimation improvements 
	achievable through power adjustment alone. ZEM exhibits the steepest degradation, declining from 
	20.7\% at $K = 8$ to 8.7\% at $K = 20$, since the zero-NMSE 
	fixed-point target becomes increasingly unattainable under 
	heavy contamination, causing the iteration to converge to a 
	suboptimal equilibrium. For $K/\tau_p \geq 4$, Exhaustive 
	yields marginally higher NMSE than both Proposed and FP due 
	to insufficient grid resolution ($G^{K/\tau_p}$ candidates), 
	confirming that gradient-based methods achieve superior 
	solution quality when the combinatorial search space grows 
	beyond tractable discretization. QoS feasibility, determined by fundamental channel conditions and bandwidth limitations as established in Section~\ref{subsubsection:array_size}, remains within $\pm 1$ percentage point across all methods at each operating point.
	
	The right panel reveals that power reduction increases 
	monotonically with $K$ for all optimization-based schemes, 
	a trend opposite to the NMSE gain behavior. 
	Proposed~($\omega > 0$) achieves the largest reduction at 
	every load, ranging from 38.9\% at $K = 8$ to 59.5\% at 
	$K = 20$, exceeding FP by 5.0 to 14.7 percentage points 
	and Proposed~($\omega = 0$) by 3.8 to 11.1 percentage 
	points. This widening gap occurs because the power regularization
	in~\eqref{eq:pilot_obj} operates with greater effect as contamination
	intensifies: $g_l^{\mathrm{self}}$ in~\eqref{eq:g_self} saturates while
	$g_l^{\mathrm{cross}}$ in~\eqref{eq:g_cross} remains active, so the explicit penalty reduces energy without meaningful NMSE degradation. ZEM maintains 
	a relatively flat power reduction around 35--40\% across 
	all loads, reflecting the static nature of its zero-NMSE 
	target that does not adapt to the evolving interference 
	structure. These results indicate that the proposed framework 
	provides increasing energy efficiency gains precisely in 
	the high-load regime where pilot power budgets are most 
	constrained, complementing the NMSE improvements observed 
	at moderate loads.
	
	Fig.~\ref{fig:runtime_gap} compares the computational cost 
	of the optimization-based pilot power allocation schemes 
	across network loads. The proposed BCD scheme with 
	$\omega = 0$ achieves up to $3\times$ 
	speedup over FP at moderate loads ($K/\tau_p \leq 2$), 
	where the scalar Newton subproblems with analytical 
	gradient and Hessian from Lemma~\ref{lem:grad_hess} 
	operate more efficiently than the $K$-dimensional 
	auxiliary variable updates in FP's quadratic transform 
	at $\mathcal{O}(K^2 M N^3 / \tau_p)$ per iteration. As 
	$K/\tau_p$ increases beyond three, the sequential 
	per-UE structure of BCD and the batch update structure 
	of FP yield comparable runtime performance, with both 
	schemes maintaining execution within sub-second scales 
	suitable for periodic reconfiguration on the slow fading 
	timescale. The power-aware variant incurs moderate 
	additional overhead from the regularization-dependent 
	gradient evaluation and stricter convergence 
	requirements. ZEM achieves the lowest runtime across all 
	loads owing to its closed-form per-UE updates at 
	$\mathcal{O}(KMN^3)$ per iteration, though this 
	computational advantage accompanies up to 22 percentage 
	points lower NMSE gain relative to the proposed scheme 
	as observed in Fig.~\ref{fig:nmse_power_vs_K}. 
	Exhaustive search, operating at 
	$\mathcal{O}(\tau_p G^{K/\tau_p} 
	|{\mathcal{M}}| N^3)$, is over two orders of magnitude slower than both BCD and FP across 
	all evaluated loads, confirming its impracticality for 
	online deployment. The proposed BCD framework thus 
	provides a favorable balance between near-optimal 
	estimation quality, tunable energy efficiency through 
	the Pareto mechanism in Fig.~\ref{fig:pareto}, and 
	computational tractability across the practical 
	operating range of CF mMIMO networks supporting URLLC 
	traffic.
	
	\begin{figure}[tb]
		\centering
		\includegraphics[scale=0.45]{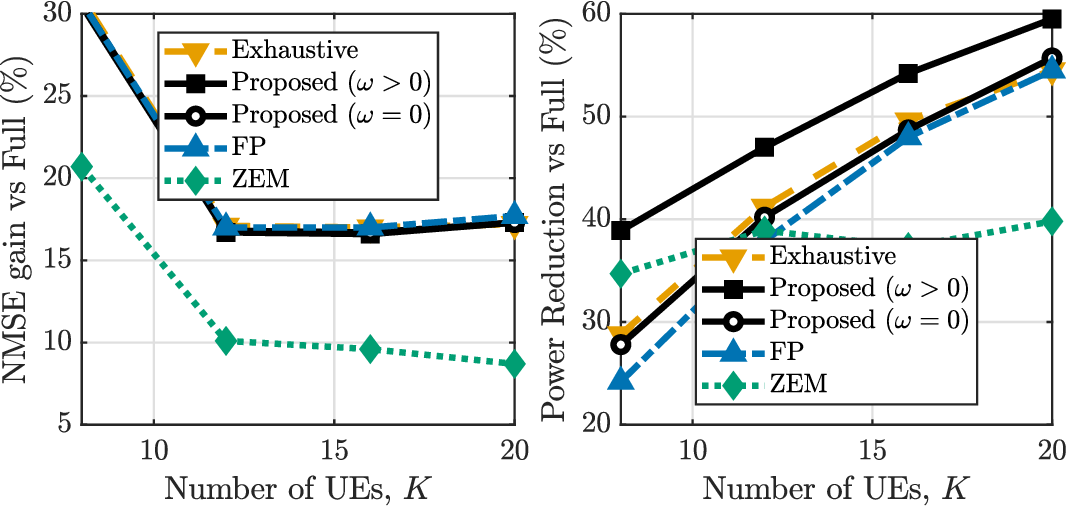}
		\caption{Impact of network load on pilot power allocation 
			performance: (left)~NMSE gain relative to full power transmission and 
			(right)~total pilot power reduction with $M=16$, $N=4$, $\tau_p=4$, and $W_m=4\lambda$. 
		}
		\label{fig:nmse_power_vs_K}
	\end{figure}
	
	\begin{figure}[tb]
		\centering
		\includegraphics[scale=0.4]{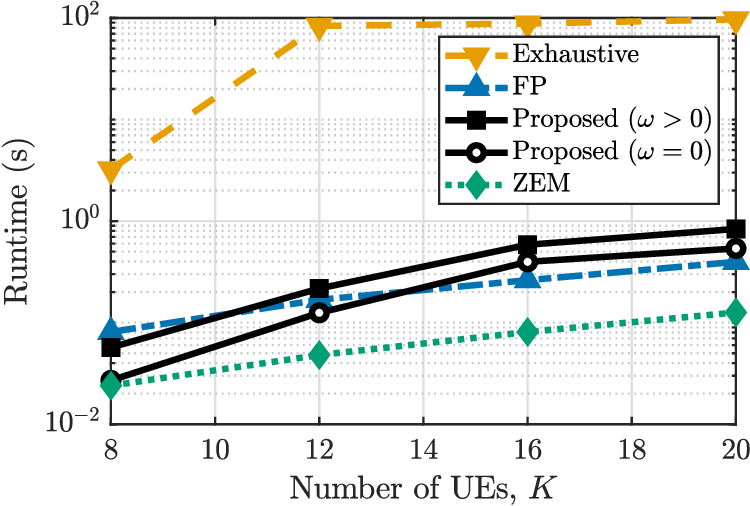}
		\caption{Computational cost of different pilot allocation schemes across 
			network loads with $M=16$, $N=4$, $\tau_p=4$, and $W_m=4\lambda$.}
		\label{fig:runtime_gap}
	\end{figure}
	\vspace{-4mm}
	\section{Conclusion}
	In this paper, we develop an AO framework for joint FA positioning and pilot power allocation in CF mMIMO networks serving URLLC traffic under FBL rate constraints and MMPP-driven bursty arrivals. An eigenstructure-based NMSE-SINR mapping translates queueing-level service requirements into physical-layer estimation targets, enabling decoupled optimization of antenna positions via SCA and pilot powers via BCD. The position optimizer achieves NMSE reduction comparable to global search at an order-of-magnitude lower cost, while the power allocator provides a tunable estimation-quality-energy tradeoff through auto-calibrated regularization that exploits the near-flat NMSE landscape near the optimum. Both stages execute within sub-second-scale runtimes suitable for periodic reconfiguration on the slow fading timescale. The results demonstrate that FA technology, coupled with estimation-aware resource optimization, provides meaningful gains in channel estimation quality and pilot
	energy efficiency within existing CF mMIMO architectures without modifying the TDD protocol or fronthaul signaling. Extending the framework to incorporate
	joint pilot sequence assignment, multi-antenna UEs with fluid antenna capability, and adaptive reconfiguration under mobility constitutes a natural direction for future work.
	\vspace{-3mm}
	\appendices
	\vspace{-3mm}
	\section{Proof of Lemma~\ref{lem:trace_approx}}
	\label{app:trace_approx}
	
	We establish the trace product approximation by analyzing the eigenstructure of the LMMSE estimator. Consider the pilot-contamination-free case $\mathcal{P}_k = \{k\}$; the extension to $|\mathcal{P}_k| > 1$ follows analogously with modified effective SNR.
	
	Let $\mathbf{R}_{k,m} = \mathbf{U} \boldsymbol{\Lambda} \mathbf{U}^\mathsf{H}$ be the eigendecomposition with $\boldsymbol{\Lambda} = \mathrm{diag}(\lambda_1, \ldots, \lambda_N)$ and $\lambda_1 \geq \cdots \geq \lambda_N > 0$. From \eqref{eq:psi_matrix}, the observation covariance becomes $\boldsymbol{\Psi}_{k,m} = \tau_p p_k^{\mathrm{p}} \beta_{k,m} \mathbf{R}_{k,m} + \sigma^2 \mathbf{I}_N$, which shares the same eigenvectors as $\mathbf{R}_{k,m}$. Substituting into \eqref{eq:phi_matrix} and simplifying yields the channel estimate covariance in diagonal form:
	\begin{equation}
		\boldsymbol{\Phi}_{k,m} = \beta_{k,m} \mathbf{U} \, \mathrm{diag}\left( \frac{\rho_{k,m} \lambda_i^2}{1 + \rho_{k,m} \lambda_i} \right)_{i=1}^N \mathbf{U}^\mathsf{H},
		\label{eq:phi_eigen_app}
	\end{equation}
	where $\rho_{k,m} = \tau_p p_k^{\mathrm{p}} \beta_{k,m} / \sigma^2$ is the effective pilot SNR.
	
	The trace products follow directly from \eqref{eq:phi_eigen_app}:
	\begin{align}
		\mathrm{tr}(\boldsymbol{\Phi}_{k,m}) &= \beta_{k,m} \sum_{i=1}^{N} \frac{\rho_{k,m} \lambda_i^2}{1 + \rho_{k,m} \lambda_i}, \label{eq:tr_phi_app} \\
		\mathrm{tr}(\boldsymbol{\Phi}_{k,m} \mathbf{R}_{k,m}) &= \beta_{k,m} \sum_{i=1}^{N} \frac{\rho_{k,m} \lambda_i^3}{1 + \rho_{k,m} \lambda_i}. \label{eq:tr_phi_R_app}
	\end{align}
	Dividing \eqref{eq:tr_phi_R_app} by \eqref{eq:tr_phi_app} gives the exact self-alignment factor $\eta_{k,m} = \mathrm{tr}(\boldsymbol{\Phi}_{k,m} \mathbf{R}_{k,m}) / \mathrm{tr}(\boldsymbol{\Phi}_{k,m})$. In contrast, the static factor from Definition~\ref{def:alignment_factors} is $\bar{\eta}_{k,m} = \sum_i \lambda_i^2 / \sum_i \lambda_i$.
	
	To quantify their difference, define the weight $w_i = \rho_{k,m} \lambda_i / (1 + \rho_{k,m} \lambda_i) \in (0,1)$, which approaches unity as $\rho_{k,m} \to \infty$. The exact factor can be written as $\eta_{k,m} = \sum_i w_i \lambda_i^2 / \sum_i w_i \lambda_i$, revealing that $\eta_{k,m} \to \bar{\eta}_{k,m}$ in the high-SNR limit since all weights converge to one.
	
	For finite SNR, defining $S_n = \sum_i \lambda_i^n$ and
	$D_n = \sum_i \lambda_i^n/(1+\rho_{k,m}\lambda_i)$,
	algebraic manipulation yields
	$\eta_{k,m} - \bar{\eta}_{k,m} = (D_1 S_2 - D_2 S_1)/(W_1 S_1)$, where $W_n = S_n - D_n$. Applying $D_n \leq S_n/(1+\rho_{k,m}\lambda_{\min})$ and
	standard trace inequalities~\cite{horn2012matrix}:
	\begin{equation}
		|\eta_{k,m} - \bar{\eta}_{k,m}| \leq
		\frac{(\lambda_{\max}-\lambda_{\min})^2}
		{4\bar{\lambda}(1+\rho_{k,m}\lambda_{\min})},
	\end{equation}
	where $\bar{\lambda} = \mathrm{tr}(\mathbf{R}_{k,m})/N$. Expressing the eigenvalue spread in terms of the condition number $\kappa_{k,m} = \lambda_{\max}/\lambda_{\min}$ gives the stated bound $\mathcal{O}((\kappa_{k,m} - 1)^2 / \rho_{k,m})$. The cross-alignment bound follows analogously using
	$\mathrm{tr}(\mathbf{A}\mathbf{B})\leq
	\|\mathbf{B}\|_2\mathrm{tr}(\mathbf{A})$ for positive
	semidefinite~$\mathbf{A}$, yielding
	$\mathcal{O}((\kappa_{k,m}-1)/\rho_{k,m})$. Under pilot
	contamination, $\rho_{k,m}$ is replaced by
	$\rho_{k,m}^{\mathrm{eff}}$ from~\eqref{eq:rho_eff}. \hfill $\blacksquare$
	
	\vspace{-3mm}
	\section{Proof of Proposition~\ref{prop:nmse_variation}}
	\label{app:nmse_variation}
	
	We derive the NMSE variation bound for UC clusters with bounded large-scale fading spread.
	
	Under uniform pilot power and $\mathbf{R}_{k,m}\approx
	\mathbf{R}_k$ for all $m\in\mathcal{M}_k$, the NMSE
	from~\eqref{eq:nmse_explicit} becomes
	$\mathrm{NMSE}_{k,m} = \frac{1}{N}\sum_{i=1}^{N}
	(1+\rho_{k,m}\lambda_i)^{-1}$ with
	$\rho_{k,m}=\tau_p p_k^{\mathrm{p}}\beta_{k,m}/\sigma^2$.
	Applying the mean value theorem to $f_i(\rho) =
	(1+\rho\lambda_i)^{-1}$ between $\rho_{k,m}$ and
	$\rho_{k,m'}$, bounding $|f_i'(\xi_i)|\leq
	\lambda_i(1+\rho_{\min}\lambda_i)^{-2}$, and using
	$|\rho_{k,m}-\rho_{k,m'}|\leq
	(\tau_p p_k^{\mathrm{p}}/\sigma^2)
	(\beta_k^{\max}-\beta_k^{\min})$ with
	$\sum_i \lambda_i = N$ yields
	\begin{equation}
		|\mathrm{NMSE}_{k,m}-\mathrm{NMSE}_{k,m'}| \leq
		\frac{\tau_p p_k^{\mathrm{p}} \Delta\beta_k}
		{\sigma^2(1+\rho_{\min})^2}.
	\end{equation}
	Substituting $\rho_{\min} = \tau_p
	p_k^{\mathrm{p}}\beta_k^{\min}/\sigma^2$ and simplifying
	gives~\eqref{eq:delta_k} which completes the proof. \hfill $\blacksquare$
	\vspace{-3mm}
	\balance
	\bibliographystyle{IEEEtran}
	\bibliography{FAS}
\end{document}